\documentclass[lettersize,journal]{IEEEtran}

\usepackage{amsmath,amssymb,amsfonts}
\usepackage{amsthm}
\usepackage{tabularx}
\usepackage{array}
\usepackage[caption=false,font=normalsize,labelfont=sf,textfont=sf]{subfig}
\usepackage{textcomp}
\usepackage{stfloats}
\usepackage{url}
\usepackage{verbatim}
\usepackage{graphicx}
\usepackage{cite}
\usepackage{bm}
\usepackage{subcaption} 
\usepackage{xcolor}
\usepackage[most]{tcolorbox}
\usepackage{dsfont}

\usepackage{amsmath}
\usepackage[ruled,vlined,linesnumbered]{algorithm2e}
\SetKwComment{Comment}{// }{}
\newtheorem{theo}{Theorem}
\newtheorem{corollary}{Corollary}

\newtheorem{assumption}{Assumption}

\newtheorem{lemma}{Lemma}

\newtheorem{remark}{Remark}

\usepackage{lscape}

\usepackage{booktabs}
\usepackage{multirow}

\tcbset{
  resultbox/.style={
    enhanced,
    breakable,
    boxrule=0.4pt,
    arc=0.8mm,
    outer arc=0.8mm,
    left=0.7mm,
    right=0.7mm,
    top=0.35mm,
    bottom=0.35mm,
    before skip=0.3\baselineskip,
    after skip=0.3\baselineskip,
    fontupper=\normalsize,
  }
}

\tcolorboxenvironment{theo}{
  resultbox,
  colback=blue!3,
  colframe=blue!55!black,
  borderline west={1.4pt}{0pt}{blue!65!black},
}

\tcolorboxenvironment{corollary}{
  resultbox,
  colback=cyan!3,
  colframe=cyan!45!black,
  borderline west={1.4pt}{0pt}{cyan!60!black},
}

\tcolorboxenvironment{lemma}{
  resultbox,
  colback=teal!3,
  colframe=teal!45!black,
  borderline west={1.4pt}{0pt}{teal!60!black},
}

\tcolorboxenvironment{prop}{
  resultbox,
  colback=teal!3,
  colframe=teal!45!black,
  borderline west={1.4pt}{0pt}{teal!60!black},
}

\tcolorboxenvironment{remark}{
  resultbox,
  colback=black!2,
  colframe=black!35,
  coltext=black,
  borderline west={1.2pt}{0pt}{black!45},
}

\tcolorboxenvironment{example}{
  resultbox,
  colback=purple!3,
  colframe=purple!45!black,
  borderline west={1.2pt}{0pt}{purple!60!black},
}

\begin{document}
\setlength\textfloatsep{8pt}
\setlength\intextsep{8pt}
\title{MM-LMPC: Multi-Modal Learning Model Predictive Control via Mode-Specific Terminal Design and Bandit-Based Exploration}

\author{Wataru Hashimoto, Kazumune Hashimoto, and Masako Kishida
\thanks{{Wataru Hashimoto and Kazumune Hashimoto are with the Graduate School of Engineering, The University of Osaka, Suita, Japan (e-mail: hashimoto@is.eei.eng.osaka-u.ac.jp, hashimoto@eei.eng.osaka-u.ac.jp). Masako Kishida is with the Institute of Systems and Information Engineering,
University of Tsukuba, Tsukuba, Ibaraki, Japan
(e-mail: kishida@cs.tsukuba.ac.jp). The corresponding author is Wataru Hashimoto.
This work is supported by JST CREST JPMJCR201, JST ACT-X JPMJAX23CK, and JSPS KAKENHI Grant 21K14184, and 22KK0155.
}}
}



\maketitle

\begin{abstract}
Learning Model Predictive Control (LMPC) improves iterative control tasks by
using previous executions to construct the terminal constraint and terminal cost
of the MPC problem. Although effective, this reuse of past trajectories can make
LMPC sensitive to the initial data. In particular, LMPC may repeatedly exploit
stored trajectories with favorable cost-to-go values while insufficiently
exploring alternative route patterns that could yield lower cost after further
improvement.
To address this issue, we propose Multi-Modal LMPC (MM-LMPC). The proposed
framework clusters past trajectories into motion modes, constructs a
mode-specific LMPC controller for each mode, and uses an LCB-based
meta-controller to select which mode-specific controller to execute at each
iteration. Mode information is incorporated into the terminal constraint and
terminal cost through two designs. The hard-constrained design uses
mode-specific terminal constraints and terminal costs constructed from the data
associated with each mode. The soft-regularized design retains a shared
terminal constraint while adding membership-based penalties to the terminal
cost. These designs reduce the bias caused by pooling all trajectories into a
single terminal memory while retaining the recursive feasibility and stability
structure of LMPC.
Our theoretical analysis shows that both designs preserve recursive feasibility
and closed-loop stability. For the hard-constrained design, we further establish
mode-wise cost convergence, asymptotic best-mode performance, and a logarithmic
cumulative regret bound under the LCB rule. Simulations on
multi-route obstacle-avoidance tasks show that MM-LMPC improves exploration and
achieves lower costs than standard LMPC.
\end{abstract}

\begin{IEEEkeywords}
Model predictive control, iterative tasks, exploration and exploitation, multi-modal control.
\end{IEEEkeywords}

\section{Introduction}
\IEEEPARstart{M}{odel} Predictive Control (MPC) is a widely used control strategy that determines control inputs by repeatedly solving a finite-horizon optimal control problem at each sampling instant based on a predictive model of the system dynamics \cite{MPC}. Its ability to explicitly handle system constraints and multivariate systems has made MPC a powerful tool in various engineering domains, from process control \cite{MPCengineering} to autonomous systems such as robotics and self-driving vehicles \cite{MPCAGV1,MPCLocomotion}.
 However, since MPC determines the control input by solving an optimal control problem over a relatively short prediction horizon, its decisions may deviate from the true infinite-horizon optimal solution. This can lead to high-cost suboptimal performance, particularly in scenarios where long-term effects and delayed consequences play a significant role in achieving the control objectives.

To address this problem, Ugo Rosolia and Francesco Borrelli 
 proposed Learning Model Predictive Control (LMPC) \cite{iterative1}, which repeatedly applies control with MPC to iterative tasks while leveraging state and input trajectories from previous iterations to improve control performance. In their approach, terminal constraints and terminal cost functions are progressively updated using data from successful past iterations, thereby ensuring recursive feasibility of the optimization problem, stability of the closed-loop system, and non-increasing iteration costs under suitable assumptions.
 
However, a key limitation of the LMPC framework is its strong dependence on the trajectories collected in early iterations.
Since both the sampled safe set and the terminal cost are constructed from states visited in previous successful trials, the finite-horizon optimizer is biased toward terminal states with favorable cost-to-go values in the stored data.
Consequently, even when the initial dataset contains feasible trajectories corresponding to multiple qualitatively different routes, trajectories belonging to modes with initially large cost-to-go values may be rarely selected as terminal candidates and may therefore remain under-explored.
For example, in a navigation task with obstacles, even if feasible trajectories are initially provided on both sides of an obstacle, LMPC may keep refining the side with the smaller initial cost-to-go and fail to discover that the other side can eventually yield a lower-cost route. A naive alternative is to run independent LMPC processes from several or all
initial trajectories and select the best result. However, this requires
repeated task executions for multiple LMPC processes, which becomes costly as
the number of initial trajectories increases.

To overcome this limitation, we propose Multi-Modal Learning Model Predictive Control (MM-LMPC), a framework that systematically explores and exploits multiple solution modes while preserving the recursive feasibility and stability properties of LMPC. Here, a mode refers to a qualitatively distinct way of accomplishing the same task. Examples include different routes around obstacles and different grasp/contact patterns in robotic manipulation. 

MM-LMPC consists of a two-level control architecture together with a
mode-identification and terminal-memory update module. The mode-identification
module associates stored trajectories with modes and uses this information to
construct and update the terminal ingredients associated with each mode.
Based on this mode information, the upper-level LCB-based multi-armed-bandit
meta-controller selects a mode by balancing exploitation of well-performing
modes with exploration of under-tested ones. The lower-level mode-specific
LMPC controller then executes the selected mode in a receding-horizon manner.
After each execution, the resulting closed-loop trajectory is processed by
the mode-identification module, and the stored trajectory data, mode
information, and terminal ingredients are updated for subsequent iterations.

For the lower-level controller, we consider two terminal designs that differ
in how trajectory data are shared across modes. The \emph{hard-constrained}
formulation uses only mode-assigned trajectories to construct mode-specific
terminal constraints and costs, thereby enforcing strict mode separation.
In contrast, the \emph{soft-regularized} formulation retains the shared
sampled safe set while adding membership-based terminal penalties. This keeps
trajectories collected under other modes available as terminal candidates
while encouraging the optimizer to select terminal states associated with the
selected mode.

Our theoretical analysis establishes recursive feasibility and asymptotic
stability for both formulations. We also derive iteration-wise cost bounds accounting for LCB exploration
and, in the soft formulation, the membership penalty.
For the hard-constrained formulation, we show that the realized
closed-loop cost approaches the best limiting cost among the modes
infinitely often. Under an additional summability assumption, we further
show that the cumulative regret relative to the best limiting cost among
the modes grows at most logarithmically with the number of iterations. Simulations on a
minimum-time reach--avoid problem for a Dubins car show that MM-LMPC explores
initially unfavorable routes and finds lower-cost trajectories than standard
LMPC.

\textbf{Related work:}

\emph{Learning from repeated executions and LMPC:}
Iterative learning control (ILC) improves tracking performance in
repetitive tasks by updating the control input using errors observed in
previous executions \cite{ILC1}. To incorporate feedback during each
execution and explicitly enforce state and input constraints, ILC has
also been combined with model predictive control (MPC)
\cite{ILMPC1,ILMPC3,ILMPC4}. Recent data-driven extensions develop
robust and model-free learning laws
\cite{TCybDataComp2022,TCybIndirectILC2024}, but these methods are
primarily formulated for tracking a prescribed reference trajectory.
Reference-free Learning Model Predictive Control (LMPC) instead
constructs a sampled terminal constraint and an empirical terminal cost
directly from successful closed-loop trajectories \cite{iterative1}.
Under suitable assumptions, this construction guarantees recursive
feasibility and closed-loop stability and yields non-increasing
iteration costs. Its convergence and optimality properties have also
been analyzed in \cite{UgoOpt}. Subsequent LMPC extensions have
addressed stochastic systems \cite{iterative4}, unknown dynamics
\cite{self2}, distributed tasks \cite{iterative5}, and alternative
certified terminal constructions \cite{NNLMPC}.

\emph{Learning terminal ingredients of MPC:}
RL and ADP have been used to approximate long-horizon value functions for MPC terminal costs \cite{AVITerminalMPC2023,VITerminalError2024}, while other approaches learn terminal regions or construct terminal ingredients directly from input--output data \cite{RLTerminalSet2026,DataTerminalMPC2021}. In contrast, MM-LMPC constructs the terminal constraint and terminal cost directly from stored trajectories, without requiring a separate learning or optimization procedure for the terminal ingredients.

\emph{Representing multiple solution modes:}
Within the LMPC literature, multimodality has also been addressed for
systems with switching physical dynamics, where historical data are
used to construct local predictive models and sampled safe sets for
different dynamical modes \cite{MMLMPC}. In contrast, the present work
considers multimodality arising from multiple qualitatively distinct
feasible solutions to a single nonconvex task, while the system
dynamics remain fixed.

Related trajectory-optimization and MPC methods have explicitly
considered multiple locally distinct solutions. Topology-aware planners
generate and optimize candidate trajectories belonging to different
homology or homotopy classes
\cite{RosmannTopology2017,deGrootTopology2025}, while multi-modal MPC
solves multiple maneuver-specific trajectory-optimization problems in
parallel and selects among the resulting locally optimal plans
\cite{AdajaniaMMPC2022}. These methods generate and compare alternative solutions within an online planning episode. MM-LMPC instead preserves mode-associated experience across complete task executions and progressively updates the corresponding terminal information, thereby enabling different solution modes to be refined over successive iterations.

A preliminary version was presented at the 2026 American Control Conference (ACC) \cite{ACC}. The current paper substantially extends it by introducing a soft-regularized terminal design, developing a unified theoretical analysis, and evaluating the proposed method on a more complex reach--avoid scenario.

\section{Problem Formulation}\label{sec:problem}

We consider a discrete-time nonlinear system
\begin{equation}\label{eq:system}
    x_{t+1} = f(x_t, u_t), \quad x_t \in \mathbb{R}^n, \; u_t \in \mathbb{R}^m,
\end{equation}
subject to state and input constraints
\begin{equation}
    x_t \in \mathcal{X}, \quad u_t \in \mathcal{U}.
\end{equation}
We consider an iterative control task in which each execution starts from the
fixed initial state $x_0=x_S$ and aims to reach the equilibrium state $x_F$.
Feasible trajectories may belong to multiple qualitatively distinct modes,
for example when the state constraint set $\mathcal{X}$ is nonconvex.

Let $\mathcal{T}$ denote the set of complete feasible trajectories from
$x_S$ to $x_F$, and let $\mathcal{M}=\{1,\ldots,M\}$ denote a finite set
of solution modes. Each mode $m\in\mathcal{M}$ is associated with a subset
\[
\mathcal{T}_m\subseteq\mathcal{T},
\]
consisting of trajectories that share a common task-relevant qualitative
pattern. Examples include different routes or homotopy patterns in obstacle
avoidance and different contact patterns in manipulation.

The performance of each execution is evaluated by the accumulated stage cost. The
optimal control problem for a single execution is given by
\begin{equation}
\begin{aligned}
    \min_{\{u_t\}_{t=0}^{\infty}} \quad
        & \sum_{t=0}^{\infty} h(x_t,u_t) \\
    \mathrm{s.t.}\quad
        & x_{t+1}=f(x_t,u_t), \\
        & x_t\in\mathcal{X},
          \quad u_t\in\mathcal{U},
          \quad \forall t\geq 0,\\
        & x_0=x_S .
\end{aligned}
\label{eq:infinite_horizon_ocp}
\end{equation}
Here, $\{x_t\}_{t=0}^{\infty}$ denotes the state sequence generated from
$x_0=x_S$ by~\eqref{eq:system}, $\{u_t\}_{t=0}^{\infty}$ denotes the
corresponding input sequence, and $h$ is the stage cost function that evaluates
the control performance.

Since solving~\eqref{eq:infinite_horizon_ocp} directly is generally difficult,
our objective is to design an iterative control method that repeatedly
executes the task from $x_S$ to $x_F$ and uses data from previous executions
to improve the control solution, as measured by the objective
in~\eqref{eq:infinite_horizon_ocp}, while ensuring constraint satisfaction
and successful arrival at $x_F$ in every execution.
We make the following assumptions on the system and the stage cost function
$h$, which are standard in the MPC literature~\cite{MPC,iterative1}.

\begin{assumption}
\label{assum:system}
The system dynamics $f(\cdot,\cdot)$ are continuous.
The state and input constraint sets $\mathcal{X}$ and $\mathcal{U}$ are compact.
\end{assumption}

\begin{assumption}
\label{assum:stagecost}
The stage cost function $h:\mathcal{X}\times\mathcal{U}\to\mathbb{R}_{\ge 0}$ is continuous and satisfies
\begin{align}
h(x,u)>0, \qquad \forall x\in\mathcal{X}\setminus\{x_F\},\; u\in\mathcal{U},
\end{align}
where the final state $x_F$ is assumed to be a feasible equilibrium of the unforced system~\eqref{eq:system}, i.e., $f(x_F,0)=x_F$.
Moreover, the function $h$ satisfies
\begin{align}
h(x_F,0) &= 0,\\
h(x_F,u) &> 0, \qquad \forall u\in\mathcal{U}\setminus\{0\}.
\end{align}
\end{assumption}
Under Assumptions~\ref{assum:system} and~\ref{assum:stagecost}, any feasible solution of~\eqref{eq:infinite_horizon_ocp} with a finite objective value satisfies $\lim_{t\to\infty}x_t=x_F$ and $\lim_{t\to\infty}u_t=0$.

\section{Review of Learning Model Predictive Control}
\label{sec:review}

In this section, we review the Learning Model Predictive Control
(LMPC) framework~\cite{iterative1}, which forms the basis of our method.
At each iteration, LMPC solves a finite-horizon MPC problem whose
terminal set and terminal cost are constructed from successful previous
executions. We then discuss a limitation of standard LMPC
that motivates MM-LMPC.

In LMPC, the task is executed repeatedly over iterations $j=1,2,\dots$ using a finite-horizon MPC.
Consistent with~\eqref{eq:infinite_horizon_ocp}, each closed-loop execution starts from the same initial state, i.e.,
\begin{equation}
    x_0^j=x_S,\qquad j=1,2,\dots .
\end{equation}
At iteration $j$, a feasible closed-loop trajectory
\[
\{x_0^j,u_0^j,x_1^j,u_1^j,\dots,x_{T_j-1}^j,u_{T_j-1}^j,x_{T_j}^j\}
\]
is obtained, where $T_j$ denotes the time to reach the final state $x_F$.
From all successful previous iterations, LMPC constructs the terminal set as
\begin{equation}
\label{eq:SS}
    \mathcal{SS}_j
    =
    \bigcup_{i\in\mathcal{I}_j}
    \bigcup_{t=0}^{T_i}
    \{x_t^i\},
\end{equation}
where $\mathcal{I}_j$ is the set of indices of iterations that successfully completed the task before iteration $j$.
For each $x\in\mathcal{SS}_j$, LMPC defines the terminal cost as the minimal cost-to-go among previous visits:
\begin{equation}
\label{eq:Q}
    Q_j(x)
    =
    \begin{cases}
    \displaystyle
    \min_{(i,t)\in\mathcal{F}^j(x)}
    \sum_{k=t}^{T_i-1} h(x_k^i,u_k^i),
    & \text{if } x\in\mathcal{SS}_j,\\[2ex]
    +\infty,
    & \text{otherwise},
    \end{cases}
\end{equation}
where
\begin{equation}
\label{eq:F}
    \mathcal{F}^j(x)
    =
    \{(i,t)\mid i\in\mathcal{I}_j,\; t\in\{0,\dots,T_i\},\; x_t^i=x\}.
\end{equation}
With the above definitions of terminal set and cost, at time $t$ in iteration $j$, LMPC solves the finite-horizon optimal control problem:
\begin{subequations}
\label{eq:mpc}
\begin{align}
    \min_{\{u_{k|t}^j\}_{k=t}^{t+N-1}} \quad
        & \sum_{k=t}^{t+N-1} h(x_{k|t}^j,u_{k|t}^j)
        + Q_{j-1}(x_{t+N|t}^j) \\
    \mathrm{s.t.}\quad
        & x_{k+1|t}^j = f(x_{k|t}^j,u_{k|t}^j),
        \quad k=t,\dots,t+N-1,\\
        & x_{k|t}^j\in\mathcal{X},
        \quad u_{k|t}^j\in\mathcal{U},
        \quad k=t,\dots,t+N-1,\\
        & x_{t+N|t}^j\in\mathcal{SS}_{j-1},\\
        & x_{t|t}^j=x_t^j,
\end{align}
\end{subequations}
where $N\in\mathbb{N}$ is the prediction horizon.
After solving~\eqref{eq:mpc}, the optimal input and corresponding state trajectories are obtained as
\(\{u_{k|t}^{j,*}\}_{k=t}^{t+N-1}\) and \(\{x_{k|t}^{j,*}\}_{k=t}^{t+N}\), respectively.
Then, the first optimal control input \(u_{t|t}^{j,*}\) is applied to the system~\eqref{eq:system}, and the next state \(x_{t+1}^j\) is observed.
At the next time step, the optimization is solved again from the initial state \(x_{t+1}^j\).
This procedure is repeated at each time step, thereby implementing
receding-horizon control. During iteration \(j\), the terminal set
\(\mathcal{SS}_{j-1}\) and terminal cost \(Q_{j-1}\) remain fixed.
After completing iteration \(j\), the collected closed-loop trajectory
is used to update them to \(\mathcal{SS}_j\) and \(Q_j\), which are
then used in iteration \(j+1\).
As discussed in Section III of~\cite{iterative1}, LMPC guarantees desirable properties such as recursive feasibility and stability of the closed-loop system, and ensures that the total cost of each iteration does not increase.

However, standard LMPC can be sensitive to the trajectories available in the early iterations. Since the sampled safe set and terminal cost are constructed from successful past trajectories, the optimizer tends to favor stored states with low cost-to-go values. As a result, routes represented by high-cost initial trajectories may be revisited only rarely and receive little opportunity for improvement.

Figure~\ref{fig:example} illustrates this behavior in the reach--avoid
problem of Section~\ref{sec:sim}. The task admits eight routes defined by
passing above or below three obstacles, and the initial dataset contains a
feasible trajectory for each route (for details, see Section~\ref{sec:sim}). Nevertheless, standard LMPC primarily
refines a geometrically longer route with a relatively low initial cost,
while routes represented by higher-cost initial trajectories receive little
refinement. This behavior motivates the proposed MM-LMPC framework.

\begin{figure}[t]
    \centering
    \includegraphics[width=0.85\linewidth]{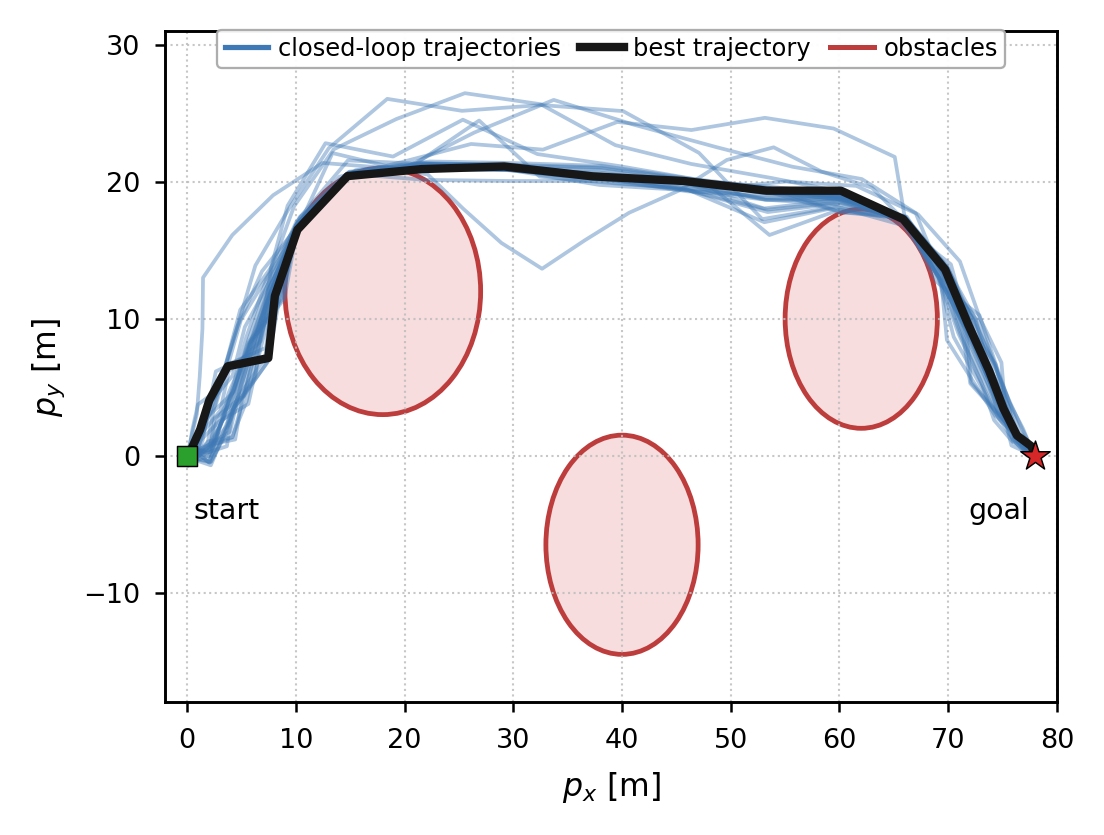}
    \caption{Closed-loop trajectories of standard LMPC in the reach--avoid example. Blue: all executions; black: the best trajectory. Despite feasible initialization of all eight routes, standard LMPC repeatedly refines one initially favorable route.}
    \label{fig:example}
\end{figure}

\section{Proposed Methodology: Multi-Modal Learning Model Predictive Control (MM-LMPC)}
\begin{figure}[t]
    \centering
    \includegraphics[width=\linewidth]{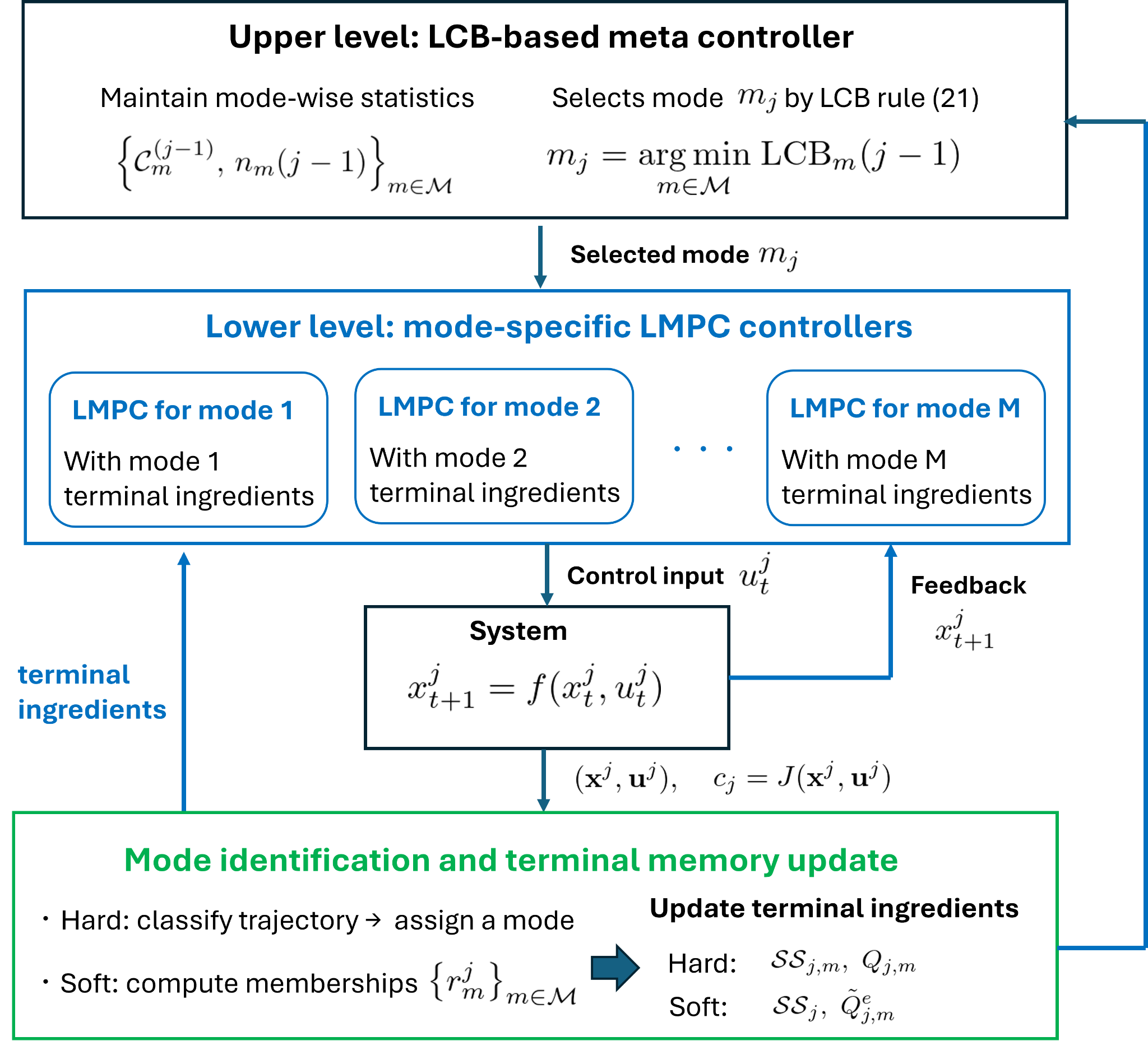}
    \caption{The proposed MM-LMPC architecture.}
    \label{fig:proposed}
\end{figure}

\begin{algorithm}[t]
\DontPrintSemicolon
\caption{Generalized MM-LMPC}
\label{alg:mm-lmpc-general}

\KwIn{$x_S$, $\mathcal{D}_0$, $J_{\max}$, $\kappa$, and
$\sigma\in\{\mathrm{hard},\mathrm{soft}\}$}

Identify ${M}_0$ and the mode information from $\mathcal{D}_0$
as described in Subsection~\ref{subsec:mode_identification}\;
Initialize the mode statistics and terminal ingredients\;

\For{$j=1$ \KwTo $J_{\max}$}{
    Compute $\hat c_m^{(j-1)}=\min\mathcal{C}_m^{(j-1)}$ and select
    $m_j\in{M}_{j-1}$ using~\eqref{eq:LCB}
    as described in Subsection~\ref{subsec:meta_controller}\;

    Set $x_0^j=x_S$ and $t\leftarrow0$\;
    \While{$x_t^j\neq x_F$}{
        Solve \eqref{eq:MPC2} or~\eqref{eq:MPC_soft_offset_np}
        for $m_j$ using the iteration-$(j-1)$ terminal ingredients
        as described in Subsection~\ref{subsec:mode_specific_lmpc}\;
        Apply $u_t^j=u_{t|t}^{j,*}$, observe $x_{t+1}^j$, and
        set $t\leftarrow t+1$\;
    }

    Store $(\mathbf{x}^j,\mathbf{u}^j)$, set
    $c_j=J(\mathbf{x}^j,\mathbf{u}^j)$, and update $\mathcal{D}_j$\;
    Update ${M}_j$ and the mode information from $\mathcal{D}_j$
    according to $\sigma$
    as described in Subsection~\ref{subsec:mode_identification}\;
    Update the mode statistics and terminal ingredients\;
}
\end{algorithm}
To address the limitation of standard LMPC discussed in the previous section, we propose
\emph{Multi-Modal LMPC} (MM-LMPC). Instead of pooling all trajectories into
a single terminal memory, MM-LMPC identifies multiple solution modes and
constructs mode-dependent terminal ingredients to enable mode-aware improvement across repeated executions.

The proposed framework combines a two-level control architecture with a
mode-identification and terminal-memory update module. This module assigns mode information to the collected trajectories and uses it to update the corresponding terminal ingredients. Using this information, the
upper-level LCB-based meta-controller determines the mode to be considered at
each global iteration, and the lower-level mode-specific LMPC controller
performs the corresponding receding-horizon control. The trajectory obtained
from the resulting closed-loop execution is subsequently used to update the
mode information and terminal memory.

For each global iteration $j$, let
$\mathbf{x}^j:=\{x_t^j\}_{t=0}^{T_j}$ and
$\mathbf{u}^j:=\{u_t^j\}_{t=0}^{T_j-1}$
denote the resulting closed-loop state and input sequences, respectively,
and define the realized cost as
$c_j:=J(\mathbf{x}^j,\mathbf{u}^j)$.
The overall architecture is illustrated in
Fig.~\ref{fig:proposed}.
Since MM-LMPC constructs its terminal ingredients from previously observed
feasible trajectories, we first impose the following assumption.

\begin{assumption}[Initial Feasible Trajectories]
\label{assum:initial_trajectory}
The initial dataset $\mathcal{D}_0$ is nonempty, and every trajectory
$(\mathbf{x}^i,\mathbf{u}^i)\in\mathcal{D}_0$ is a complete feasible
closed-loop trajectory from $x_S$ to $x_F$. In particular, it satisfies
\begin{align}
    x_0^i=x_S,\qquad x_{T_i}^i=x_F,
\end{align}
as well as the system dynamics and all state and input constraints.
\end{assumption}

The following subsections describe the components of MM-LMPC. For ease
of definition, their presentation order differs from the execution order
summarized in Algorithm~\ref{alg:mm-lmpc-general}.

\subsection{Mode Identification}\label{subsec:mode_identification}
The solution modes introduced in Section~\ref{sec:problem}, together with the association of stored trajectories with these modes, are specified using task-dependent knowledge or inferred from trajectory data using classification, clustering, or probabilistic modeling methods, such as neural-network classifiers, $k$-means, DBSCAN~\cite{DBSCAN}, or Gaussian mixture models~\cite{GMM}. For each stored trajectory $(\mathbf{x}^i,\mathbf{u}^i)$, the mode-identification step provides either a hard label or a membership vector
\begin{align}
\bm{r}_i = (r_{i,1},\dots,r_{i,M}),
\end{align}
where $r_{i,m}\in[0,1]$ represents the degree of association between trajectory
$i$ and mode $m$.

The mode labels and membership values are then used differently in the two formulations introduced in the next subsection.
The proposed framework is not tied to a particular identification method, provided that the resulting mode information satisfies the assumptions introduced in the theoretical analysis.

\subsection{Mode-Specific LMPC}\label{subsec:mode_specific_lmpc}

Given a selected mode $m$, MM-LMPC executes iteration $j$ using an LMPC controller whose mode information is incorporated through the terminal ingredients. We consider two complementary designs: the hard-constrained formulation uses a mode-specific terminal constraint and terminal cost, whereas the soft-regularized formulation retains the shared sampled safe set and introduces a mode-dependent regularization term in the terminal cost. For each iteration $j$, let $\mathcal{I}_{j}$ denote the index set of stored closed-loop trajectories available up to iteration $j$.

\subsubsection{Hard-Constrained Mode-Specific LMPC}

Using only the trajectory data associated with mode $m$, we construct a
\emph{mode-specific sampled safe set} $\mathcal{SS}_{j,m}$ and a
\emph{mode-specific value function} $Q_{j,m}(\cdot)$, defined as in
(\ref{eq:SS}) and (\ref{eq:Q}) on the subset of trajectories assigned to mode
$m$.
At time $t$ of iteration $j$, the controller for the selected mode $m$ solves
\begin{align}\label{eq:MPC2}
\min_{\{u_{k|t}\}} \quad &
\sum_{k=t}^{t+N-1} h(x_{k|t},u_{k|t}) + Q_{j-1,m}(x_{t+N|t}) \nonumber\\
\text{s.t.}\quad &
x_{k+1|t}=f(x_{k|t},u_{k|t}), \nonumber\\
& x_{k|t}\in\mathcal{X},\; u_{k|t}\in\mathcal{U}, \nonumber\\
& x_{t+N|t}\in\mathcal{SS}_{j-1,m}, \nonumber\\
& x_{t|t}=x_t^j .
\end{align}
The resulting control inputs are applied in a receding-horizon manner, with the
selected mode $m$ fixed throughout iteration $j$.

\subsubsection{Soft-Regularized Mode-Specific LMPC}
The hard-constrained formulation can be sample-inefficient because, for a selected mode \(m\), it excludes terminal samples from trajectories assigned to other modes. To alleviate this issue, the soft-regularized design retains all past samples in
the shared sampled safe set $\mathcal{SS}_{j-1}$ and biases terminal selection
through a membership-based regularization term.
For each stored trajectory $i\in\mathcal{I}_{j-1}$, let
$r_{i,m}^{(j-1)}\in[0,1]$ denote its degree of association with mode $m$ defined in Subsection~\ref{subsec:mode_identification}
at iteration $j-1$. We define
\begin{align}
\psi_{i,m}^{(j-1)}
:=
\rho\,\phi\!\left(1-r_{i,m}^{(j-1)}\right),
\label{eq:psi_def_offset_np}
\end{align}
where $\rho\geq0$ is a design parameter and $\phi(\cdot)$ is monotonically
increasing. We also define the offset term used in the terminal cost below as
\begin{align}
b_{j-1,m}
:=
\min_{i\in \mathcal{I}_{j-1}} \psi_{i,m}^{(j-1)},
\label{eq:b_def_offset_np}
\end{align}
For a terminal state $x\in\mathcal{SS}_{j-1}$, let
\begin{align}
\mathcal{F}^{j-1}(x)
:=
\{(i,t)\mid i\in\mathcal{I}_{j-1},\; x_t^i=x\}
\end{align}
be the set of stored trajectory-time pairs containing $x$. The mode-dependent
terminal cost is defined as
\begin{align}
&\widetilde{Q}^{e}_{j-1,m}(x)\notag \\
&:=
\min_{(i,t)\in\mathcal{F}^{j-1}(x)}
\left(
\sum_{k=t}^{T_i-1} h(x_k^i,u_k^i)
+
\psi_{i,m}^{(j-1)}
-
b_{j-1,m}
\right).
\label{eq:Qe_tilde_def_np}
\end{align}

Given the mode $m$ selected at the beginning of iteration $j$ and
fixed throughout that iteration, the soft-regularized controller solves
the following problem at each time $t$:
\begin{align}\label{eq:MPC_soft_offset_np}
\min_{\{u^j_{k|t}\}} \quad &
\sum_{k=t}^{t+N-1} h(x^j_{k|t},u^j_{k|t})
+ \widetilde{Q}^{e}_{j-1,m}(x^j_{t+N|t}) \nonumber\\
\text{s.t.}\quad &
x^j_{k+1|t}=f(x^j_{k|t},u^j_{k|t}), \nonumber\\
& x^j_{k|t}\in\mathcal{X},\; u^j_{k|t}\in\mathcal{U}, \nonumber\\
& x^j_{t+N|t}\in \mathcal{SS}_{j-1}, \nonumber\\
& x^j_{t|t}=x^j_t .
\end{align}
For both designs, we
denote by $c_{m,k}$ the total cost of the closed-loop trajectory obtained when
mode $m$ is executed for the $k$-th time.
\subsection{Meta-Controller for Mode Selection}\label{subsec:meta_controller}

The choice of which mode to execute at the beginning of iteration $j+1$ is posed as a \emph{multi-armed bandit} (MAB) problem~\cite{bandit1,bandit2}, where each mode is treated as an arm.
The meta-controller should exploit modes that have already achieved low costs, while still occasionally testing modes that have been selected only a few times.
To achieve this balance, we adopt a \emph{Lower Confidence Bound} (LCB) rule.

Let $M_j$ denote the set of modes represented in the stored data up to iteration $j$.
For each mode $m \in M_j$, let $n_m(j)$ denote the number of times mode $m$ has been selected up to iteration $j$, and let $\mathcal{C}_m^{(j)}$ denote the set of realized costs recorded for mode $m$ up to iteration $j$.
We then define
\begin{align}
\hat c_m^{(j)} := \min\!\left(\mathcal{C}_m^{(j)}\right),
\end{align}
where $\hat c_m^{(j)}$ represents the best cost statistic currently associated with mode $m$.
With these quantities, the mode selected at iteration $j+1$ is given by
\begin{align}\label{eq:LCB}
m_{j+1}
=
\arg\min_{m \in M_j}
\left(
\hat c_m^{(j)}
-
\kappa
{\scriptstyle\sqrt{\tfrac{\log(j+1)}{\max\{1,n_m(j)\}}}}
\right),
\end{align}
where $\kappa > 0$ controls the exploration--exploitation trade-off.

In this rule, the first term favors modes that have already produced low realized costs.
The second term lowers the score of modes with small selection counts, thereby encouraging the controller to revisit under-tested modes.
Thus, the LCB rule selects the mode with the most promising optimistic score, rather than simply choosing the mode with the lowest cost observed so far.
For both the hard and soft formulations, the meta-controller selects the mode according to the same bandit-based rule.

\subsection{Overall Procedure}

The complete MM-LMPC procedure is summarized in Algorithm~1.
Here, $J_{\max}$ denotes the number of global iterations, and
$\sigma \in \{\mathrm{hard},\mathrm{soft}\}$ specifies the terminal design.
Before the iterative executions begin, the initial feasible dataset
$\mathcal{D}_0$ is used to identify the available modes and to initialize
the corresponding mode statistics and terminal ingredients.

At each global iteration $j$, the meta-controller first selects a mode
$m_j$ according to the LCB rule using the performance statistics
accumulated from previous executions.
The selected mode is then fixed throughout the iteration, and the
corresponding mode-specific LMPC problem is solved repeatedly in a
receding-horizon manner until the task is completed.
After the resulting closed-loop trajectory and its realized cost are
recorded, the dataset is augmented, and the mode information, mode
statistics, and terminal ingredients are updated for the next iteration.

\section{Theoretical Analysis}

In this section, we establish the main theoretical results of the proposed
MM-LMPC framework. We first prove recursive feasibility and asymptotic stability for both
formulations, and then analyze the mode-wise evolution of the cost and derive
iteration-wise closed-loop cost bounds under LCB-based mode selection.
For the hard-constrained formulation, we further establish mode-wise cost
convergence, asymptotic best-mode performance, and logarithmic regret.

For a fixed iteration $j$ and selected mode $m$, let $V_{j,m}(x)$ denote the
optimal value of the corresponding finite-horizon problem initialized at
state $x$. Along the closed-loop trajectory of iteration $j$, define
\[
V_t^\ast := V_{j,m}(x_t^j)
=
\sum_{k=t}^{t+N-1}
h\!\left(x_{k|t}^{j,\ast},u_{k|t}^{j,\ast}\right)
+
\mathcal{Q}_{j-1,m}\!\left(x_{t+N|t}^{j,\ast}\right),
\]
where
\[
\mathcal{Q}_{j-1,m}
=
\begin{cases}
Q_{j-1,m},
& \text{hard-constrained},\\
\widetilde Q^{e}_{j-1,m},
& \text{soft-regularized}.
\end{cases}
\]
\begin{theo}[Recursive Feasibility and Asymptotic Stability]
\label{thm:feasibility}
Consider the proposed MM-LMPC scheme with
the hard terminal-constraint formulation~\eqref{eq:MPC2} or the
soft-penalized formulation~\eqref{eq:MPC_soft_offset_np}.
Suppose Assumptions~\ref{assum:system}, \ref{assum:stagecost}, and
\ref{assum:initial_trajectory} hold.
Then, for all iterations \(j\ge1\) and all times \(t\ge0\), the
corresponding finite-horizon optimal control problem is recursively feasible.
Moreover, for each fixed iteration \(j\) and selected mode \(m\), the
equilibrium \(x_F\) is asymptotically stable and the resulting closed-loop
trajectory satisfies
\[
x_t^j\to x_F,\qquad u_t^j\to0
\quad\text{as }t\to\infty.
\]
\end{theo}

\begin{proof}
For notational simplicity, we omit the iteration superscript $j$ from the
closed-loop state and input throughout the proof.

In the hard formulation, once a mode is selected, the controller is a standard
LMPC controller with the corresponding mode-specific sampled safe set and
terminal cost. Therefore, recursive feasibility and asymptotic stability follow
from the standard LMPC argument~\cite{iterative1}. In the soft formulation, the
terminal constraint and the stored feasible continuations are unchanged, while
only the terminal cost is modified by the membership-based penalty. Therefore,
recursive feasibility follows from the same LMPC argument, whereas
stability is established by the additional argument below.

\emph{Asymptotic stability of the soft-regularized formulation.}
We use the optimal value function $V_{j,m}$ associated with
\eqref{eq:MPC_soft_offset_np} as a Lyapunov candidate. By Assumption~\ref{assum:stagecost} and the definition of the
soft-regularized terminal cost, we have
\begin{align}
\widetilde Q^{e}_{j-1,m}(x)\ge0,
\qquad
\forall x\in\mathcal{SS}_{j-1}.
\end{align}
Hence, $V_t^*\ge0$ for every feasible state. Moreover, since each stored
feasible trajectory reaches $x_F$ and the equilibrium can be extended by
applying $u=0$, the remaining cost from $x_F$ is zero. Therefore,
\begin{align}
\widetilde Q^{e}_{j-1,m}(x_F)
&=
\min_{i\in\mathcal I_{j-1}}\psi_{i,m}^{(j-1)}-b_{j-1,m} \\
&=0 ,
\end{align}
where we used the definition of $b_{j-1,m}$. Since the zero input keeps the
system at $x_F$, it follows that $V_{j,m}(x_F)=0$. If $x_t\neq x_F$, then the
first stage cost is positive by Assumption~\ref{assum:stagecost}, while all
remaining cost terms are nonnegative. Thus, $V_t^*>0$ for $x_t\neq x_F$.

We next show that $V_t^*$ decreases along the closed-loop trajectory. Let
\[
\{u^{*}_{k|t}\}_{k=t}^{t+N-1},
\qquad
\{x^{*}_{k|t}\}_{k=t}^{t+N}
\]
be an optimal solution of~\eqref{eq:MPC_soft_offset_np} at time $t$. 
Since
$x^{*}_{t+N|t}\in\mathcal{SS}_{j-1}$, choose a stored trajectory-time
pair $(\bar i,\ell)\in\mathcal{F}^{j-1}(x^{*}_{t+N|t})$ that attains the
minimum in the definition of
$\widetilde Q^{e}_{j-1,m}(x^{*}_{t+N|t})$. Then,
\begin{align}
x_{\ell}^{\bar i}
=
x^{*}_{t+N|t}
\end{align}
and
\begin{align}
\widetilde Q^{e}_{j-1,m}(x_{\ell}^{\bar i})
=
\sum_{k=\ell}^{T_{\bar i}-1}
h(x_k^{\bar i},u_k^{\bar i})
+\psi_{\bar i,m}-b_{j-1,m}.
\end{align}
Using the same stored trajectory to evaluate the terminal cost at the successor
stored state gives
\begin{align}
\widetilde Q^{e}_{j-1,m}(x^{\bar i}_{\ell+1})
&\le
\sum_{k=\ell+1}^{T_{\bar i}-1} h(x_k^{\bar i},u_k^{\bar i})
+\psi_{\bar i,m}^{(j-1)}-b_{j-1,m} \nonumber\\
&=
\sum_{k=\ell}^{T_{\bar i}-1} h(x_k^{\bar i},u_k^{\bar i})
+\psi_{\bar i,m}^{(j-1)}-b_{j-1,m}
-h(x_{\ell}^{\bar i},u_{\ell}^{\bar i}) \nonumber\\
&=
\widetilde Q^{e}_{j-1,m}(x_{\ell}^{\bar i})
-h(x_{\ell}^{\bar i},u_{\ell}^{\bar i}).
\label{eq:soft_Q_decrease}
\end{align}
If $x_{\ell}^{\bar i}=x_F$, the same inequality follows from the equilibrium
extension $f(x_F,0)=x_F$ and $h(x_F,0)=0$.

Now evaluate the optimization problem at time $t+1$ using the standard shifted
candidate: the tail of the optimal sequence at time $t$, followed by the stored
continuation input $u_{\ell}^{\bar i}$. This gives
\begin{align}
V_{t+1}^*
&\le
\sum_{k=t+1}^{t+N-1} h(x^{*}_{k|t},u^{*}_{k|t})
+ h(x_{\ell}^{\bar i},u_{\ell}^{\bar i})
+ \widetilde Q^{e}_{j-1,m}(x_{\ell+1}^{\bar i}) \nonumber\\
&\le
\sum_{k=t+1}^{t+N-1} h(x^{*}_{k|t},u^{*}_{k|t})
+ \widetilde Q^{e}_{j-1,m}(x_{\ell}^{\bar i}) \nonumber\\
&=
V_t^* - h(x_t,u_t).
\label{eq:V_decrease_soft}
\end{align}
Consequently, by Assumption~\ref{assum:stagecost}, $V_t^*$ is strictly
decreasing whenever $x_t\neq x_F$ and remains constant at the equilibrium.

Summing~\eqref{eq:V_decrease_soft} from $t=0$ to $T-1$ yields
\begin{align}
\sum_{t=0}^{T-1} h(x_t,u_t)
\le
V_0^* - V_T^*
\le
V_0^* .
\end{align}
Since $h(x_t,u_t)\ge0$, the partial sums are monotone and bounded, and hence
\begin{align}
\sum_{t=0}^{\infty} h(x_t,u_t)<\infty .
\end{align}
Therefore, $h(x_t,u_t)\to0$ as $t\to\infty$. By Assumptions~\ref{assum:system} and \ref{assum:stagecost}, this implies
\begin{align}
(x_t,u_t)\to(x_F,0),
\end{align}
or equivalently, $x_t\to x_F$ and $u_t\to 0$. Furthermore,
continuity of $V_{j,m}$ can be established as in the standard MPC
stability analysis~\cite{MPC,iterative1}. Therefore, the positive
definiteness and continuity of $V_{j,m}$, together with the decrease
condition~\eqref{eq:V_decrease_soft}, establish Lyapunov stability of $x_F$. Combined with
the above convergence result, this proves asymptotic stability of the
soft-regularized formulation.
\end{proof}

\begin{remark}
The present analysis is restricted to the nominal disturbance-free setting and extending MM-LMPC to systems with disturbances is outside the scope of this paper. In particular, the exact terminal-state matching used here is generally not preserved under disturbances. Nevertheless, related extensions have been studied in the LMPC literature for uncertain and stochastic settings, and similar ideas could potentially be incorporated into the present mode-dependent framework; see, e.g.,~\cite{iterative3, iterative4, self2}.
\end{remark}

To analyze the mode-wise evolution of the iteration cost, we now introduce the additional assumptions used in the subsequent results.
\begin{assumption}[Finiteness of Modes]
\label{assum:finiteness}
Let $M_j$ denote the set of modes discovered up to iteration $j$.
We assume that there exists a finite set $M_\infty$ such that
\[
M_j = M_\infty
\qquad
\text{for all sufficiently large } j.
\]
That is, the set of discovered modes eventually stabilizes and contains only finitely many elements.
\end{assumption}

\begin{assumption}[Eventual Consistency of Hard Mode Classification]
\label{assum:hard_consistency}
There exists a finite iteration index $J_c\in\mathbb{N}$ such that,
for every mode $m\in M_\infty$ and every iteration $j>J_c$, if the
hard MM-LMPC controller is executed with selected mode $m$, then the
resulting closed-loop trajectory is classified into the same mode $m$.
\end{assumption}

Assumption~\ref{assum:finiteness} means that the task has only finitely
many relevant motion modes. This excludes cases where new qualitatively
different modes keep appearing indefinitely.
Assumption~\ref{assum:hard_consistency} ensures that, after finitely
many iterations, the hard mode-specific terminal ingredients are
updated consistently.

The following lemma characterizes the mode-wise evolution of the
iteration cost. 
\begin{lemma}[Intra-Mode Cost Evolution]
\label{lem:intra_mode_cost_evolution}
Suppose Assumptions~\ref{assum:system}--\ref{assum:finiteness} hold.
For the hard formulation, suppose additionally
Assumption~\ref{assum:hard_consistency} holds.
Fix any mode $m \in M_\infty$, let $j_{m,k}$ denote the global iteration
index of the $k$-th execution of mode $m$, and define
\begin{align}
c_{m,k}
:=
\sum_{t=0}^{T_{j_{m,k}}-1}
h\!\left(x_t^{j_{m,k}},u_t^{j_{m,k}}\right).
\end{align}
Then the following statements hold.

\textit{(i) Hard formulation.}
For every $k\ge1$ such that $j_{m,k}>J_c$,
\begin{align}
c_{m,k+1} \le c_{m,k}.
\label{eq:hard_intra_mode_nonincreasing_unified}
\end{align}

\textit{(ii) Soft formulation.}
For every $k \ge 1$ and every $\ell \in \{1,\dots,k\}$,
\begin{align}
c_{m,k+1}
\le
c_{m,\ell}
+
\psi_{j_{m,\ell},m}^{(j_{m,k+1}-1)}
-
b_{j_{m,k+1}-1,m}.
\label{eq:soft_intra_mode_bounded_increase_unified}
\end{align}

\end{lemma}

\begin{proof}
We prove the two statements separately.

\textit{(i) Hard formulation.}
For any $k\ge1$ such that $j_{m,k}>J_c$,
Assumption~\ref{assum:hard_consistency} ensures that the trajectory
generated at iteration $j_{m,k}$ is classified into mode $m$.
Therefore, this trajectory is incorporated into the mode-specific sampled
safe set and terminal cost used at iteration $j_{m,k+1}$.
Hence, the standard LMPC monotonicity argument in
Theorem~2 of~\cite{iterative1} applies mode-wise:
\begin{align}
c_{m,k+1} \le c_{m,k}.
\end{align}

\textit{(ii) Soft formulation.}
Fix $k \ge 1$ and an arbitrary $\ell \in \{1,\dots,k\}$.
Set
\begin{align}
j^\ell := j_{m,\ell},
\qquad
j^+ := j_{m,k+1}.
\end{align}
Let $V_t^{j^+}$ denote the optimal value of
\eqref{eq:MPC_soft_offset_np} at time $t$ of iteration $j^+$.

To compare the cost of the trajectory generated at iteration $j^\ell$
with the optimal value at iteration $j^+$, we separate its first $N$
stage costs from the remaining cost-to-go and include the penalty offset
used at iteration $j^+$:
\begin{align}
&c_{m,\ell}
+\psi_{j^\ell,m}^{(j^+-1)}
-b_{j^+-1,m}
\notag\\
&=
\sum_{t=0}^{N-1} h(x_t^{j^\ell},u_t^{j^\ell})
\notag\\
&\qquad+
\Bigl(
\sum_{t=N}^{T_{j^{\ell}}-1}
h(x_t^{j^\ell},u_t^{j^\ell})
+\psi_{j^\ell,m}^{(j^+-1)}
-b_{j^+-1,m}
\Bigr).
\label{eq:split_cost_unified}
\end{align}
Since $x_N^{j^\ell}\in\mathcal{SS}_{j^+-1}$ and
$\widetilde Q^{e}_{j^+-1,m}$ is defined by minimizing over stored
trajectories in \eqref{eq:Qe_tilde_def_np}, evaluating it on the
stored trajectory $j^\ell$ with index $t=N$ yields
\begin{align}
\widetilde Q^{e}_{j^+-1,m}(x_N^{j^\ell})
\le
\sum_{t=N}^{T_{j^\ell}-1}
h(x_t^{j^\ell},u_t^{j^\ell})
+\psi_{j^\ell,m}^{(j^+-1)}
-b_{j^+-1,m}.
\label{eq:Qe_prev_bound_unified}
\end{align}
Combining \eqref{eq:split_cost_unified} and
\eqref{eq:Qe_prev_bound_unified} gives
\begin{align}
&c_{m,\ell}
+\psi_{j^\ell,m}^{(j^+-1)}
-b_{j^+-1,m}
\notag\\
&\qquad\quad\ge
\sum_{t=0}^{N-1} h(x_t^{j^\ell},u_t^{j^\ell})
+
\widetilde Q^{e}_{j^+-1,m}(x_N^{j^\ell}).
\label{eq:bound_candidate_cost_unified}
\end{align}
The right-hand side of
\eqref{eq:bound_candidate_cost_unified} is the objective value of
\eqref{eq:MPC_soft_offset_np} evaluated using the first $N$ steps of
the stored trajectory
$(\mathbf{x}^{j^\ell},\mathbf{u}^{j^\ell})$
as a feasible candidate at time $0$ of iteration $j^+$.
Therefore, by optimality of $V_0^{j^+}$,
\begin{align}
c_{m,\ell}
+\psi_{j^\ell,m}^{(j^+-1)}
-b_{j^+-1,m}
\ge
V_0^{j^+}.
\label{eq:36m_final_unified}
\end{align}

By Theorem~\ref{thm:feasibility}, within the fixed iteration $j^+$,
the optimal value satisfies the one-step decrease
\begin{align}
V_{t+1}^{j^+}
\le
V_t^{j^+}-h(x_t^{j^+},u_t^{j^+}),
\qquad
\forall t\ge0.
\label{eq:Vdec_unified}
\end{align}
Iterating \eqref{eq:Vdec_unified} yields, for any $T\ge1$,
\begin{align}
V_0^{j^+}
\ge
\sum_{t=0}^{T-1}h(x_t^{j^+},u_t^{j^+})
+
V_T^{j^+}.
\label{eq:telescoping_unified}
\end{align}
From $V_T^{j^+}\ge0$ and taking
$T=T_{j_{m,k+1}}$, we obtain
\begin{align}
V_0^{j^+}
\ge
\sum_{t=0}^{T_{j_{m,k+1}}-1}
h(x_t^{j^+},u_t^{j^+})
=
c_{m,k+1}.
\label{eq:39m_final_unified}
\end{align}

Combining \eqref{eq:36m_final_unified} and
\eqref{eq:39m_final_unified} gives
\begin{align}
c_{m,k+1}
\le
V_0^{j^+}
\le
c_{m,\ell}
+
\psi_{j^\ell,m}^{(j^+-1)}
-
b_{j^+-1,m}.
\end{align}
Since $\ell\in\{1,\dots,k\}$ was arbitrary,
\eqref{eq:soft_intra_mode_bounded_increase_unified} follows.

\end{proof}

For the hard-constrained formulation, the eventual intra-mode
non-increasing cost property in
Lemma~\ref{lem:intra_mode_cost_evolution}(i) leads to the following
asymptotic best-mode performance result.

\begin{theo}[Asymptotic Best-Mode Performance]
\label{thm:best_mode_performance}
Consider the hard-constrained formulation under the assumptions of
Lemma~\ref{lem:intra_mode_cost_evolution}(i), with mode selection given
by the LCB rule~\eqref{eq:LCB}. Then, for every $m\in M_\infty$, the
mode-wise cost sequence $\{c_{m,k}\}_{k\ge1}$ converges to a finite
limit $c_m^*$. Moreover,
\begin{align}
\liminf_{j\to\infty}c_j
=
c^*,
\qquad
c^*
:=
\min_{m\in M_\infty}c_m^*.
\label{eq:asymptotic_best_mode}
\end{align}
\end{theo}

\begin{proof}
Under the LCB rule, every mode in $M_\infty$ is selected infinitely
often~\cite{bandit1,bandit2}. Indeed, if a mode were selected only
finitely many times, its selection count would remain bounded, and the
negative exploration term in~\eqref{eq:LCB} would eventually force its
reselection. Hence, $n_m(j)\to\infty$ and thus $j_{m,k}\to\infty$ for
every $m\in M_\infty$.

Consequently, for each mode $m$, there exists a finite $K_m$ such that
$j_{m,k}>J_c$ for all $k\ge K_m$. By
Lemma~\ref{lem:intra_mode_cost_evolution}(i), the tail sequence
$\{c_{m,k}\}_{k\ge K_m}$ is non-increasing. Since the costs are
nonnegative, it converges to a finite limit $c_m^*$, and the full
sequence converges to the same limit.

Let $m^*\in\arg\min_{m\in M_\infty}c_m^*$. Since mode $m^*$ is selected
infinitely often, $c_{j_{m^*,k}}=c_{m^*,k}\to c^*$, which implies
$\liminf_{j\to\infty}c_j\le c^*$.
On the other hand, the eventual non-increasing property implies
$c_{m,k}\ge c_m^*\ge c^*$ for all sufficiently large $k$. Since
$M_\infty$ is finite, it follows that $c_j\ge c^*$ for all sufficiently
large $j$, and hence $\liminf_{j\to\infty}c_j\ge c^*$. Combining the two
inequalities proves~\eqref{eq:asymptotic_best_mode}.
\end{proof}

While the preceding theorem establishes mode-wise cost convergence and
asymptotic best-mode performance, the following results analyze finite-time
behavior. For the hard formulation, we strengthen
Assumption~\ref{assum:hard_consistency} by taking $J_c=0$.

\begin{theo}[Cost Bound under LCB Selection]
\label{thm:cost_increase_bound}
Define
\[
\begin{aligned}
\hat c_m^{(j-1)}
&:=\min\mathcal C_m^{(j-1)},\\
m^\star
&\in\arg\min_{m\in M_{j-1}}\hat c_m^{(j-1)},\\
c_{\mathrm{best}}^{(j-1)}
&:=\hat c_{m^\star}^{(j-1)}.
\end{aligned}
\]
Suppose that the LCB policy selects mode $m_j$ at iteration $j$, and
let $c_j$ be the resulting iteration cost. Define
\begin{equation}
\Delta_j
:=
\begin{cases}
0,
& \text{hard-constrained},\\
\psi_{i_j^\star,m_j}^{(j-1)}-b_{j-1,m_j},
& \text{soft-regularized},
\end{cases}
\label{eq:Delta_j}
\end{equation}
where, in the soft-regularized formulation, $i_j^\star$ is a stored
trajectory index associated with mode $m_j$ that attains
$\hat c_{m_j}^{(j-1)}$. Then
\begin{align}
c_j
\le {}&
c_{\mathrm{best}}^{(j-1)}
+
\kappa
\sqrt{
\frac{\log j}
{\max\{1,n_{m_j}(j-1)\}}
}
\nonumber\\
&-
\kappa
\sqrt{
\frac{\log j}
{\max\{1,n_{m^\star}(j-1)\}}
}
+
\Delta_j .
\label{eq:lcb-cost-bound-final}
\end{align}
\end{theo}

\begin{proof}
Let $k_j$ be the unique local execution index satisfying
$j=j_{m_j,k_j}$; then $c_j=c_{m_j,k_j}$.

For the hard formulation, since Assumption~\ref{assum:hard_consistency}
holds with $J_c=0$, Lemma~\ref{lem:intra_mode_cost_evolution}(i) implies
that the mode-wise cost is non-increasing from the first execution. Hence,
\begin{align}
c_j
\le
\hat c_{m_j}^{(j-1)}.
\end{align}

For the soft formulation, choose $\ell_j^\star$ such that
$c_{m_j,\ell_j^\star}=\hat c_{m_j}^{(j-1)}$, and let
$i_j^\star:=j_{m_j,\ell_j^\star}$. Then
Lemma~\ref{lem:intra_mode_cost_evolution}(ii) gives
\begin{align}
c_j
=
c_{m_j,k_j}
&\le
c_{m_j,\ell_j^\star}
+
\psi_{i_j^\star,m_j}^{(j-1)}
-
b_{j-1,m_j}
\notag\\
&=
\hat c_{m_j}^{(j-1)}
+
\Delta_j.
\end{align}

Therefore, by the definition of $\Delta_j$, both formulations satisfy
\begin{align}
c_j
\le
\hat c_{m_j}^{(j-1)}
+
\Delta_j.
\label{eq:lcb-bound-step1-fixed}
\end{align}

Because $m_j$ is selected by the LCB rule,
\begin{align}
&\hat c_{m_j}^{(j-1)}
-
\kappa
\sqrt{
\frac{\log j}
{\max\{1,n_{m_j}(j-1)\}}
}
\notag\\
&\qquad\le
\hat c_{m^\star}^{(j-1)}
-
\kappa
\sqrt{
\frac{\log j}
{\max\{1,n_{m^\star}(j-1)\}}
}.
\end{align}
Using
$\hat c_{m^\star}^{(j-1)}=c_{\mathrm{best}}^{(j-1)}$
and rearranging yields
\begin{align}
\hat c_{m_j}^{(j-1)}
\le {}&
c_{\mathrm{best}}^{(j-1)}
+
\kappa
\sqrt{
\frac{\log j}
{\max\{1,n_{m_j}(j-1)\}}
}
\notag\\
&-
\kappa
\sqrt{
\frac{\log j}
{\max\{1,n_{m^\star}(j-1)\}}
}.
\end{align}
Combining this inequality with
\eqref{eq:lcb-bound-step1-fixed}
proves~\eqref{eq:lcb-cost-bound-final}.
\end{proof}

We now introduce an additional assumption used in the regret analysis.

\begin{assumption}[Mode-wise Improvement Summability]
\label{assum:conv_rate}
For every mode $m\in M_\infty$,
\begin{align}
\sum_{k=1}^{\infty}
\left(c_{m,k}-c_m^*\right)_+
\le
C_m
<
\infty,
\end{align}
where $c_m^*$ is defined in
Theorem~\ref{thm:best_mode_performance},
$(a)_+:=\max\{a,0\}$, and $C_m$ is a finite constant.
\end{assumption}
Assumption~\ref{assum:conv_rate} strengthens the eventual mode-wise
non-increasing property established in
Lemma~\ref{lem:intra_mode_cost_evolution}(i). It is introduced to ensure
that the cumulative transient suboptimality within each mode remains
uniformly bounded.

For the hard formulation, the eventual intra-mode non-increasing property
in Lemma~\ref{lem:intra_mode_cost_evolution}(i), together with
Assumption~\ref{assum:conv_rate}, allows us to bound both the transient
suboptimality within each mode and the regret induced by the LCB-based
mode selection rule.

\begin{theo}[Logarithmic Regret Bound]
\label{thm:regret}
Consider the hard formulation under
Assumptions~\ref{assum:system}, \ref{assum:stagecost},
\ref{assum:initial_trajectory}, \ref{assum:finiteness},
\ref{assum:hard_consistency}, and \ref{assum:conv_rate},
with $J_c=0$. Let $J_r$ be sufficiently large that
$M_j=M_\infty$ for all $j\ge J_r$ and every mode in $M_\infty$
has been executed at least once by iteration $J_r$.
Then, there exists a finite constant $C$, independent of $T$, such that
\begin{align}
R_T
:=
\sum_{j=1}^{T}(c_j-c^*)
\le
\sum_{m:\,\Delta_m>0}
\frac{4\kappa^2}{\Delta_m}\log T
+
C
=
O(\log T),
\end{align}
where
$c^*:=\min_{m\in M_\infty}c_m^*$ and
$\Delta_m:=c_m^*-c^*$.
\end{theo}

\begin{proof}
The regret accumulated before iteration $J_r$ consists of finitely many
terms and is therefore bounded above by a finite constant $C_{\rm tr}$.
For $T\ge J_r$, we have
\begin{align}
R_T
\le{}&
C_{\rm tr}
+
\underbrace{
\sum_{j=J_r}^{T}
\left(c_{m_j}^*-c^*\right)
}_{\text{(A) Suboptimal Selection Regret}}
+
\underbrace{
\sum_{j=J_r}^{T}
\left(c_j-c_{m_j}^*\right)
}_{\text{(B) Intra-Mode Cost Gap}}.
\label{eq:regret_decomposition}
\end{align}

For term (B), Assumption~\ref{assum:conv_rate} gives
\begin{align}
\sum_{j=J_r}^{T}
\left(c_j-c_{m_j}^*\right)
&\le
\sum_{m\in M_\infty}
\sum_{k=1}^{\infty}
\left(c_{m,k}-c_m^*\right)_+
\notag\\
&\le
\sum_{m\in M_\infty}C_m,
\end{align}
which is finite and independent of $T$.

It remains to bound term (A). Let
$m^*\in\arg\min_{m\in M_\infty}c_m^*$, and fix any mode $m$ with
$\Delta_m>0$. Since Assumption~\ref{assum:hard_consistency} holds with
$J_c=0$, Lemma~\ref{lem:intra_mode_cost_evolution}(i) implies
\begin{align}
\hat c_m^{(j-1)}
\ge
c_m^*.
\end{align}

Moreover, the LCB rule selects every mode infinitely often. In
particular, $m^*$ is selected infinitely often, and
$c_{m^*,k}\to c^*$. Hence, for any
$\varepsilon\in(0,\Delta_m)$, there exists a finite $J_\varepsilon$
such that
\begin{align}
\hat c_{m^*}^{(j-1)}
\le
c^*+\varepsilon,
\qquad
\forall j\ge J_\varepsilon.
\end{align}

If mode $m$ is selected at iteration $j\ge J_\varepsilon$, the LCB
rule gives
\begin{align}
\hat c_m^{(j-1)}
-
\kappa
\sqrt{
\frac{\log j}
{n_m(j-1)}
}
\le
\hat c_{m^*}^{(j-1)}
-
\kappa
\sqrt{
\frac{\log j}
{n_{m^*}(j-1)}
}.
\end{align}
Using the preceding bounds and dropping the nonpositive exploration
term associated with $m^*$, we obtain
\begin{align}
\Delta_m-\varepsilon
\le
\kappa
\sqrt{
\frac{\log j}
{n_m(j-1)}
}.
\label{eq:necessary_condition_eps}
\end{align}

Let $\tau_s$ denote the iteration at which mode $m$ is selected for the
$s$-th time. Applying
\eqref{eq:necessary_condition_eps} at $j=\tau_s$ gives, for all
sufficiently large $s$,
\begin{align}
s-1
\le
\frac{\kappa^2}{(\Delta_m-\varepsilon)^2}
\log\tau_s.
\end{align}
Consequently, there exists a finite constant $C_{0,m}$, independent of
$T$, such that
\begin{align}
n_m(T)
\le
\frac{\kappa^2}{(\Delta_m-\varepsilon)^2}
\log T
+
C_{0,m}.
\end{align}

Choosing $\varepsilon=\Delta_m/2$ yields
\begin{align}
n_m(T)
\le
\frac{4\kappa^2}{\Delta_m^2}
\log T
+
C_{0,m}.
\label{eq:pull_bound_clean}
\end{align}

Therefore,
\begin{align}
\sum_{j=J_r}^{T}
\left(c_{m_j}^*-c^*\right)
&\le
\sum_{m:\,\Delta_m>0}
n_m(T)\Delta_m
\notag\\
&\le
\sum_{m:\,\Delta_m>0}
\frac{4\kappa^2}{\Delta_m}\log T
+
\sum_{m:\,\Delta_m>0}
C_{0,m}\Delta_m.
\end{align}

Combining the bounds for terms (A) and (B) gives
\begin{align}
R_T
\le
\sum_{m:\,\Delta_m>0}
\frac{4\kappa^2}{\Delta_m}\log T
+
C,
\end{align}
where
\begin{align}
C
:=
C_{\rm tr}
+
\sum_{m\in M_\infty}C_m
+
\sum_{m:\,\Delta_m>0}C_{0,m}\Delta_m
\end{align}
is finite and independent of $T$. Enlarging $C$, if necessary, also
covers the finitely many horizons $T<J_r$. Therefore,
$R_T=O(\log T)$.
\end{proof}

\section{Simulation Study}
\label{sec:sim}

This section evaluates the proposed MM-LMPC framework on a minimum-time
reach--avoid task with multiple obstacles. We compare standard LMPC with the
hard-constrained and soft-regularized MM-LMPC designs under the same dynamics,
constraints, prediction horizon, and initial dataset. We also examine the
sensitivity of the MM-LMPC results to the LCB exploration parameter $\kappa$
and, for the soft design, the terminal-penalty weight $\rho$.

\subsection{Experimental Setup}
\label{subsec:sim_setup}

We consider a discrete-time Dubins-type vehicle model with state
$x_k=[p_{x,k}\;p_{y,k}\;v_k]^\top\in\mathbb{R}^3$ and input
$u_k=[\theta_k\;a_k]^\top\in\mathbb{R}^2$, governed by
\begin{subequations}\label{eq:sim_dubins_v14}
\begin{align}
 p_{x,k+1} &= p_{x,k}+v_k\cos\theta_k,\\
 p_{y,k+1} &= p_{y,k}+v_k\sin\theta_k,\\
 v_{k+1} &= v_k+a_k,
\end{align}
\end{subequations}
with input constraints
$-\pi/2\le \theta_k \le \pi/2$ and
$-0.8\le a_k \le 0.8$.
The initial and target states are $x_S=[0\;0\;0]^\top$ and
$x_F=[78\;0\;0]^\top$, respectively. The stage cost is one before reaching the
target and zero at the target. Hence, the realized closed-loop cost equals the
number of control inputs applied before task completion. The prediction horizon
is $N=5$.

The collision-free region is the exterior of the following three elliptical
obstacles:
$(p_{x,k}-18)^2/9^2+(p_{y,k}-12)^2/9^2\ge 1$,
$(p_{x,k}-40)^2/7^2+(p_{y,k}+6.5)^2/8^2\ge 1$, and
$(p_{x,k}-62)^2/7^2+(p_{y,k}-10)^2/8^2\ge 1$.
These obstacles induce eight route patterns according to whether a trajectory
passes above or below each obstacle. We denote the resulting mode set by
$\mathcal{M}=\{UUU,UUL,ULU,ULL,LUU,LUL,LLU,LLL\}$.
For example, $ULU$ denotes a trajectory that passes above the first
obstacle, below the second, and above the third.

For the $q$-th obstacle with center $(c_{x,q},c_{y,q})$, the route classifier
finds the trajectory sample whose horizontal coordinate is closest to
$c_{x,q}$. The $q$-th letter is set to $U$ when the corresponding vertical
coordinate is larger than $c_{y,q}$ and to $L$ otherwise. Applying this rule to
the three obstacles yields a label in $\{U,L\}^3$.

One feasible initial trajectory is prepared for each route pattern, as shown in
Fig.~\ref{fig:seeds_8modes_v14}. The corresponding costs are listed in
Table~\ref{tab:initial_demonstrations_v16}. Although the $LUL$ route is
geometrically the most direct, its initial trajectory has the largest cost.

\begin{figure}[t]
    \centering
    \includegraphics[width=0.8\linewidth]{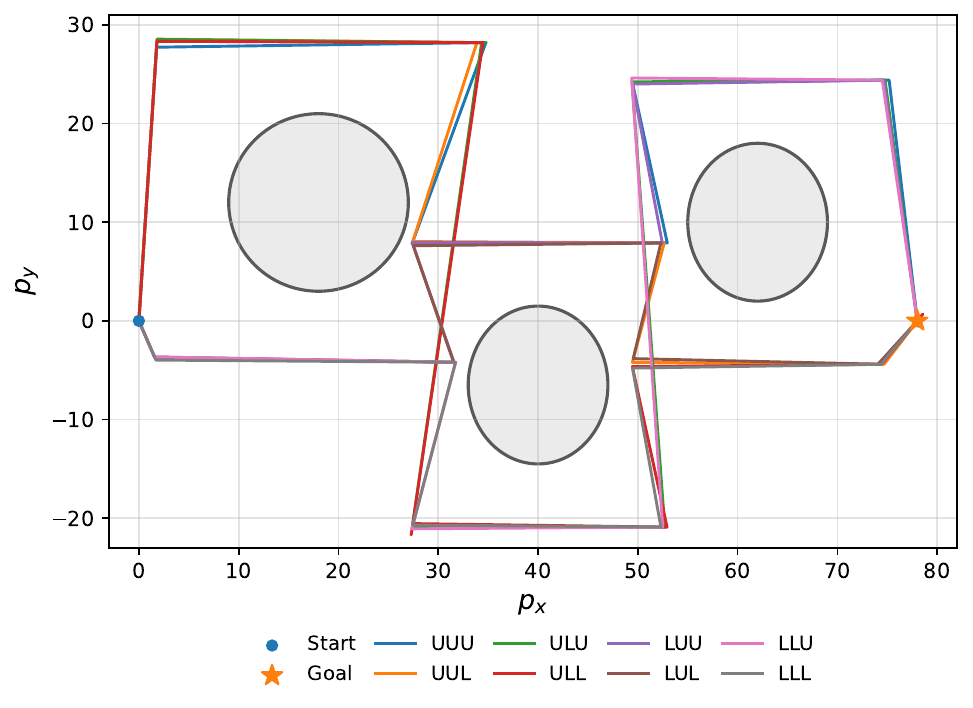}
    \caption{Initial trajectories for the eight route patterns in the
    three-obstacle reach--avoid task.}
    \label{fig:seeds_8modes_v14}
\end{figure}

\begin{table}[t]
\centering
\caption{Initial costs for the eight route modes.}
\label{tab:initial_demonstrations_v16}
\setlength{\tabcolsep}{6pt}
\renewcommand{\arraystretch}{1.0}
\begin{tabular}{cc@{\hspace{18pt}}cc}
\toprule
Route & Cost & Route & Cost \\
\midrule
$UUU$ & 111 & $LUU$ & 137 \\
$UUL$ & 121 & $LUL$ & 185 \\
$ULU$ & 135 & $LLU$ & 122 \\
$ULL$ & 108 & $LLL$ & 160 \\
\bottomrule
\end{tabular}
\end{table}

For the soft formulation, memberships are assigned directly from the route
labels. For a stored trajectory $\tau$ and mode $m$, let
$\ell(\tau),\ell(m)\in\{U,L\}^3$ denote their route labels. We define
$r_{\tau,m}=(1/3)\sum_{q=1}^{3}
\mathbf{1}\{\ell_q(\tau)=\ell_q(m)\}$,
so that $r_{\tau,m}\in\{0,1/3,2/3,1\}$. The corresponding terminal penalty is
$\psi_{\tau,m}=\rho(1-r_{\tau,m})$,
where $\rho\ge0$ determines the strength of the mode-dependent preference.
Each experiment is run for $30$ global iterations. For both MM-LMPC designs,
we evaluate $\kappa\in\{1,10,50,100\}$. For soft MM-LMPC, we additionally
evaluate $\rho\in\{0,10,30,50,150,200,300,500,800\}$.

All MPC problems were implemented using CasADi~\cite{Andersson2019}.
The implementation builds on the publicly available LMPC
code \cite{rosolia_lmpc_code}, which was extended
to implement the proposed hard-constrained and soft-regularized
MM-LMPC designs.

\subsection{Comparison and Effects of Design Parameters}
\label{subsec:sim_comparison_v14}

\subsubsection{Overall Performance Comparison}

We first briefly compare the best closed-loop costs attained by standard LMPC,
hard MM-LMPC, and soft MM-LMPC over 30 iterations. For each MM-LMPC
variant, we report the best result over the tested values of the LCB
exploration parameter $\kappa$ and, for soft MM-LMPC, the terminal-penalty
weight $\rho$. Table~\ref{tab:method_summary_v14} reports the best cost,
its first attainment iteration, and the corresponding parameter setting.

\begin{table}[t]
\centering
\caption{Best-cost comparison over 30 iterations.}
\label{tab:method_summary_v14}
\footnotesize
\renewcommand{\arraystretch}{1.08}
\setlength{\tabcolsep}{3.5pt}
\begin{tabular*}{\columnwidth}{
@{\extracolsep{\fill}}lcccc@{}
}
\toprule
Method
& \shortstack{Best\\cost}
& \shortstack{First\\iter.}
& $\kappa$
& $\rho$ \\
\midrule
Standard LMPC
& 24
& 13
& --
& -- \\
Hard MM-LMPC
& 22
& 15
& $\{50,100\}$
& -- \\
Soft MM-LMPC
& \textbf{21}
& \textbf{19}
& $10$
& $300$ \\
\bottomrule
\end{tabular*}
\end{table}

Standard LMPC attains its best cost of $24$ at iteration $13$.
As shown in Fig.~\ref{fig:example}, all of its closed-loop
trajectories follow the $UUU$ route, although initial trajectories are
available for all eight routes. Because standard LMPC uses a single
terminal memory shared across all routes, stored states on the $UUU$
route with relatively low cost-to-go values are repeatedly selected as
terminal candidates, causing the controller to continue refining the
same route.

Hard MM-LMPC attains a lower cost of $22$ at iteration $15$ for
$\kappa\in\{50,100\}$. Soft MM-LMPC attains the lowest cost of
$21$, first at iteration $19$ with $(\rho,\kappa)=(300,10)$.
All trajectories attaining cost $21$ are classified as $LUL$.
Overall, both MM-LMPC variants attain lower costs than standard LMPC,
with soft MM-LMPC achieving the lowest overall cost. As examined below,
this improvement is enabled by the ability of MM-LMPC to revisit and
refine routes other than the initially favored routes.

\subsubsection{Effects of the Design Parameters}
\label{subsec:sim_soft_v14}

We next examine how the exploration parameter $\kappa$ affects both
MM-LMPC designs and how the terminal membership penalty $\rho$ affects
soft MM-LMPC. For hard MM-LMPC, $\kappa=1$ attains a best cost of $23$
at iteration $10$, whereas $\kappa=10$ attains cost $22$ at iteration
$18$. For $\kappa\in\{50,100\}$, the same cost of $22$ is reached earlier,
at iteration $15$.

For soft MM-LMPC, the effects of $\rho$ and $\kappa$ are examined jointly.
Figure~\ref{fig:heatmap_v14} reports the best cost and its first attainment
iteration over the tested $(\rho,\kappa)$ grid. The lowest cost
of $21$ is attained for $(\rho,\kappa)\in\{(300,1),(300,10),(500,10),(800,10)\}$.
Among these settings, $(\rho,\kappa)=(300,10)$ reaches this cost first, at
iteration $19$. For $\rho=300$, $\kappa\in\{50,100\}$ attains a best cost of
$22$ within the $30$ iterations. As $\rho$ decreases toward zero, the
soft-regularized formulation approaches standard LMPC because the
membership-based terminal penalty vanishes, whereas larger values of $\rho$
enforce stronger mode preference and lead to behavior closer to hard mode
separation.

\begin{figure}[t]
    \centering
    \includegraphics[width=\linewidth]
    {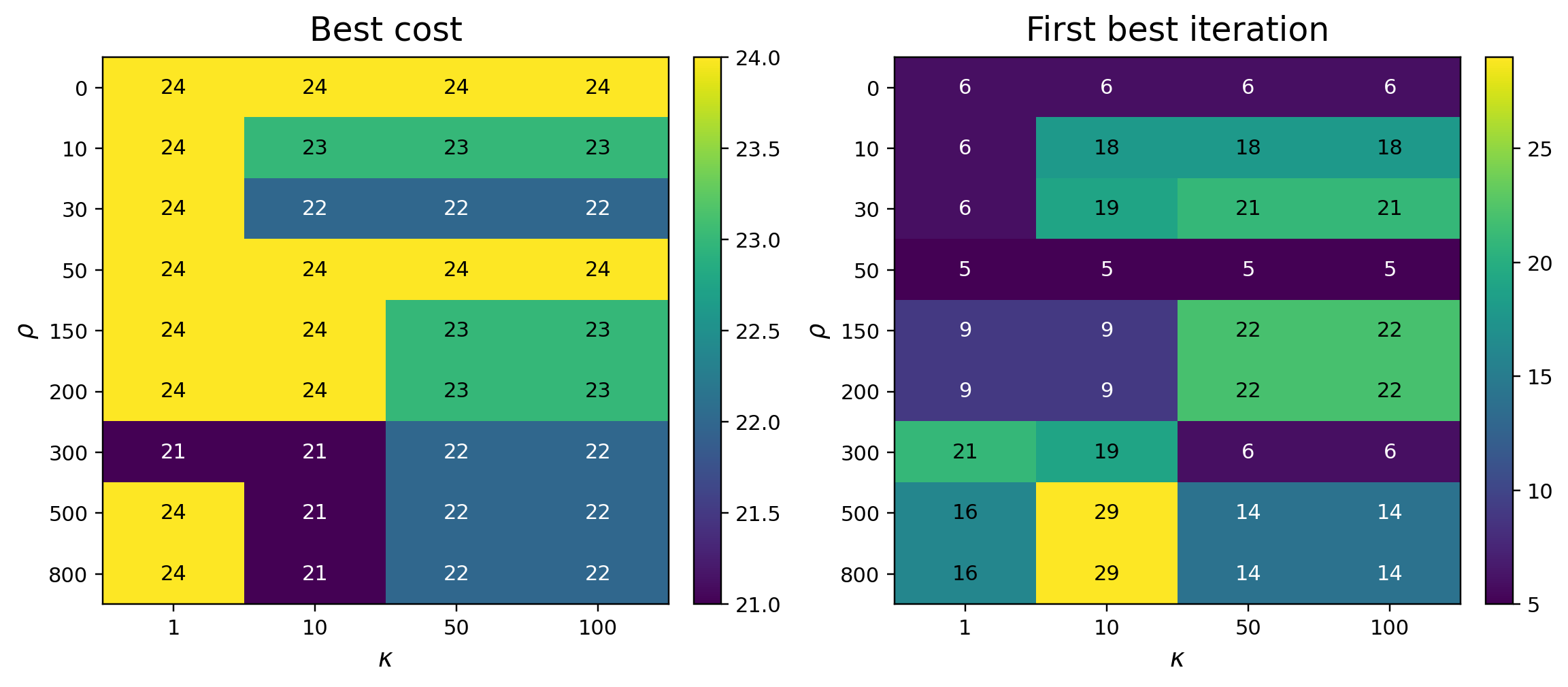}
    \caption{Performance of soft MM-LMPC over the tested
    $(\rho,\kappa)$ grid. Left: best attained closed-loop cost.
    Right: first iteration attaining that cost.}
    \label{fig:heatmap_v14}
\end{figure}

Table~\ref{tab:rho_sensitivity_v14} reports the results for different
values of $\rho$ with $\kappa=10$. Mode agreement is defined as the
fraction of iterations for which the selected mode equals the route label
of the resulting closed-loop trajectory.
\begin{table}[t]
\centering
\caption{Effect of $\rho$ on soft MM-LMPC with $\kappa=10$.}
\label{tab:rho_sensitivity_v14}
\footnotesize
\renewcommand{\arraystretch}{1.10}
\begin{tabular*}{\columnwidth}{
@{\extracolsep{\fill}}
rccc@{\hspace{10pt}}rccc@{}
}
\toprule
$\rho$
& \shortstack{Best\\cost}
& \shortstack{First\\iter.}
& \shortstack{Mode\\agreement}
& $\rho$
& \shortstack{Best\\cost}
& \shortstack{First\\iter.}
& \shortstack{Mode\\agreement} \\
\cmidrule(lr){1-4}\cmidrule(lr){5-8}
0   & 24 & 6 & 0.13
& 200 & 24 & 9 & 0.33 \\
10  & 23 & 18 & 0.13
& 300 & \textbf{21} & 19 & 0.77 \\
30  & \textbf{22} & 19 & 0.20
& 500 & \textbf{21} & 29 & 1.00 \\
50  & 24 & 5  & 0.13
& 800 & \textbf{21} & 29 & 1.00 \\
150 & 24 & 9 & 0.33
&     &    &    &      \\
\bottomrule
\end{tabular*}
\end{table}
At small $\rho$, mode agreement is low because the membership term has
little influence on the terminal-state choice. Increasing
$\rho$ generally strengthens the relation between the selected mode and the
realized route: mode agreement increases to $0.77$ at $\rho=300$ and reaches
$1.00$ for $\rho\in\{500,800\}$. The lowest cost of $21$ is attained for
$\rho\in\{300,500,800\}$, with $\rho=300$ reaching this cost earliest, at
iteration $19$. Thus, $\rho=300$ provides a favorable balance between mode
consistency and rapid discovery of a low-cost route in this experiment.

The executed trajectories in Fig.~\ref{fig:traj_rho_v14} show the
closed-loop behavior of soft MM-LMPC with $\kappa=10$ for selected
values of the terminal membership penalty $\rho$, together with the
hard MM-LMPC result for $\kappa=10$ in the bottom panel. At $\rho=10$,
the trajectories remain concentrated around the $UUU$ route, and the
resulting behavior closely resembles that of standard LMPC. At
$\rho=300$, the controller explores several route families, whereas the
larger penalty $\rho=500$ produces more mode-consistent trajectories.
For $\rho=500$, the selected mode and realized route agree in all
$30$ executions, illustrating the increasingly hard-like mode separation
induced by a large membership penalty.

\begin{figure}[t]
    \centering
    \includegraphics[width=\linewidth]
    {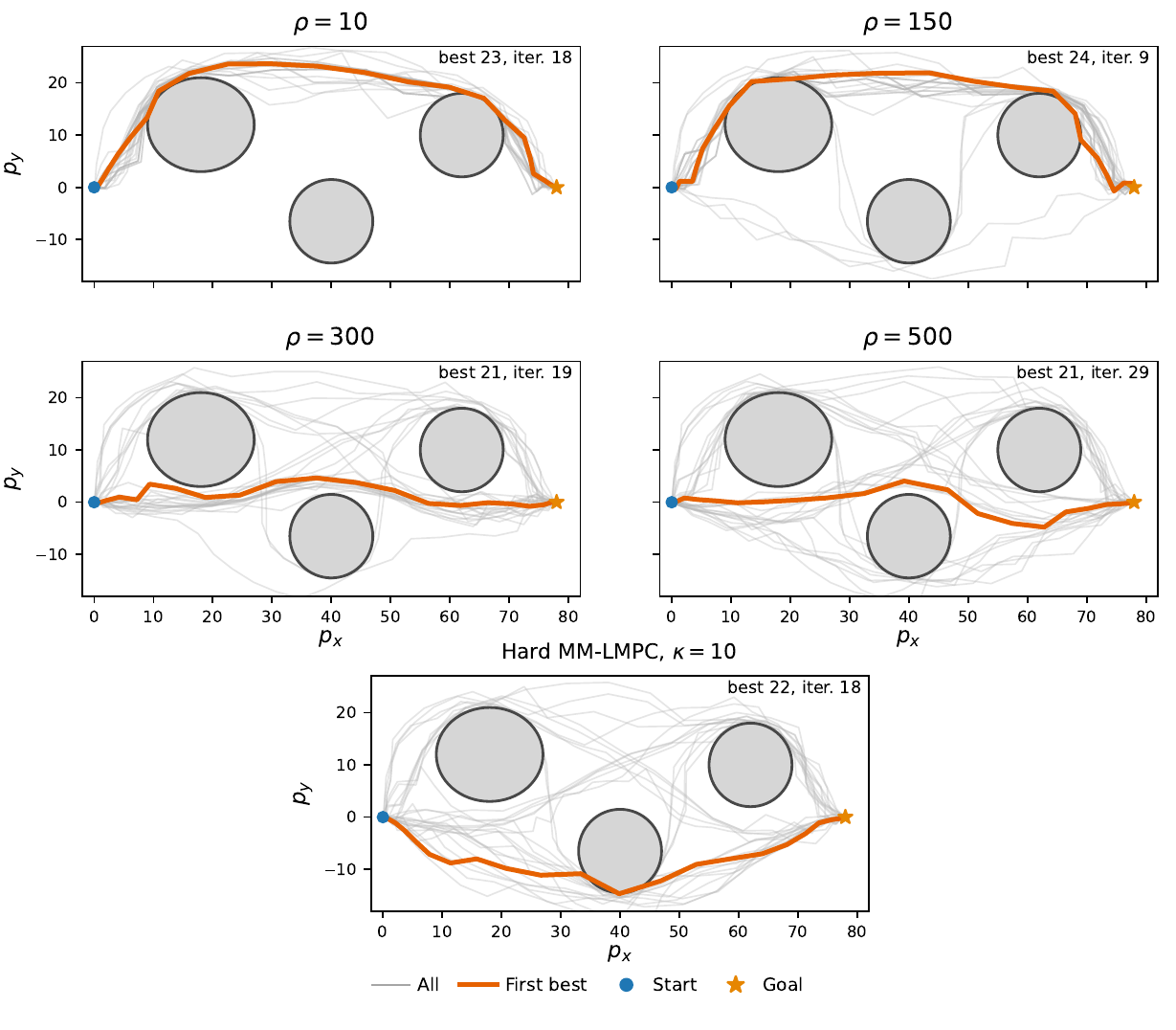}
    \caption{Executed trajectories for soft MM-LMPC at selected values
    of $\rho$ and for hard MM-LMPC, with $\kappa=10$ in all panels.
    Thin gray curves show all $30$ executions, and the solid highlighted
    curve is the first trajectory attaining the best cost in each setting.
    The hard-MM-LMPC result is shown in the bottom panel.}
    \label{fig:traj_rho_v14}
\end{figure}

Finally, Fig.~\ref{fig:traj_kappa_v14} illustrates how $\kappa$ affects
the mode-selection behavior of soft MM-LMPC at $\rho=300$. A smaller
value of $\kappa$ places greater emphasis on modes that have already
achieved low costs, so the executed trajectories remain concentrated
around the corresponding routes. As $\kappa$ increases, the exploration
term in the LCB rule becomes more influential, causing under-tested
modes to be selected more frequently and alternative routes to be
revisited more often.

\begin{figure}[t]
    \centering
    \includegraphics[width=\linewidth]
    {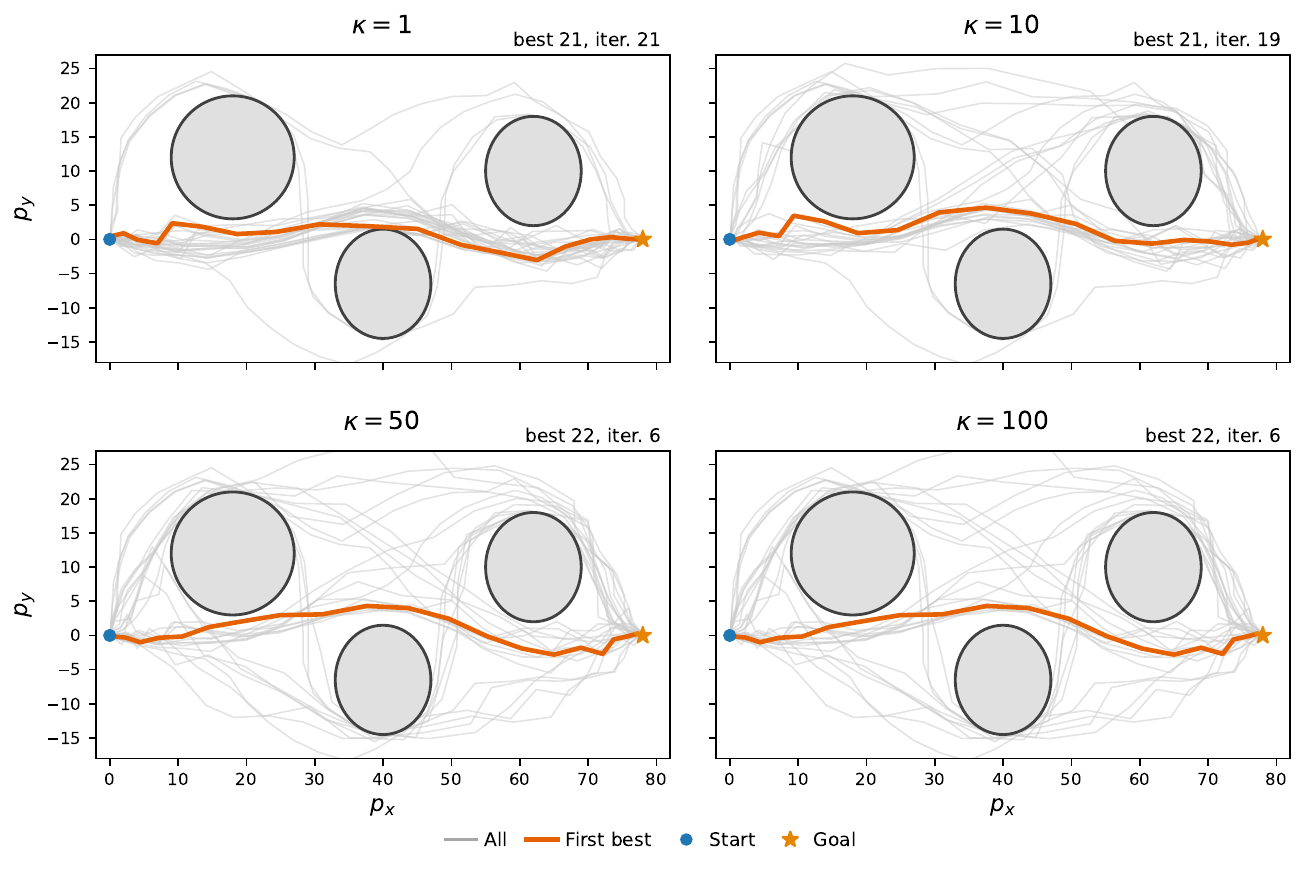}
    \caption{Executed trajectories for different values of $\kappa$
    with $\rho=300$. Thin gray curves show all $30$ executions, and the
    solid highlighted curve is the first best trajectory.}
    \label{fig:traj_kappa_v14}
\end{figure}

\section{Conclusion}

This paper presented MM-LMPC, a mode-aware extension of LMPC designed to reduce bias toward
initially favored solution modes. MM-LMPC selects a solution mode at each
iteration using an LCB-based exploration rule and incorporates the selected mode
into the LMPC terminal design. Two formulations were developed: a
hard-constrained formulation with mode-specific sampled safe sets and terminal
costs, and a soft-regularized formulation with a shared sampled safe set and
membership-based terminal penalties. We showed that both formulations preserve
recursive feasibility and closed-loop stability, characterized their mode-wise
cost evolution, and established, for the hard formulation, asymptotic best-mode
performance and an $O(\log T)$ regret bound. Simulations on minimum-time
reach--avoid tasks showed that MM-LMPC explores under-tested modes and finds
lower-cost trajectories than standard LMPC.

\bibliographystyle{IEEEtran}

\bibliography{IEEE}

\begin{thebibliography}{99}
\bibitem{rosolia2016learning}
U. Rosolia and F. Borrelli, "Learning model predictive control for iterative tasks," in \textit{2016 IEEE 55th Conference on Decision and Control (CDC)}, 2016, pp. 6498-6503. [2]

\bibitem{rosolia2017computationally}
U. Rosolia and F. Borrelli, "Learning model predictive control for iterative tasks: A data-driven control framework," \textit{IEEE Transactions on Automatic Control}, vol. 63, no. 7, pp. 1883-1896, 2017. [1]

\bibitem{rosolia2018learning}
U. Rosolia and F. Borrelli, "Learning model predictive control for linear systems: A computationally efficient approach," in \textit{2018 Annual American Control Conference (ACC)}, 2018, pp. 1-6. [3]

\bibitem{rawlings2017model}
J. B. Rawlings, D. Q. Mayne, and M. M. Diehl, \textit{Model Predictive Control: Theory, Computation, and Design}, 2nd ed. Nob Hill Publishing, 2017. [4]

\bibitem{lmpc_case_studies}
Number Analytics, "LMPC in Action: Case Studies," \textit{Number Analytics Blog}, 2023. [5]

\bibitem{hewing2020learning}
L. Hewing, A. D. Wabersich, M. N. Zeilinger, and J. Lygeros, "Learning-based model predictive control: Toward safe learning in control," \textit{Annual Review of Control, Robotics, and Autonomous Systems}, vol. 3, pp. 269-296, 2020. [6]

\bibitem{rosolia2017computationally_efficient}
U. Rosolia and F. Borrelli, "Learning Model Predictive Control for Iterative Tasks: A Computationally Efficient Approach for Linear System," \textit{arXiv preprint arXiv:1702.07064}, 2017. [7]

\bibitem{local_optima_ai_ml}
Alphanome, "Understanding Local Optima in AI/ML," \textit{Alphanome Blog}, 2023. [8]

\bibitem{local_optima_theory_of_constraints}
F. Labs, "Theory of Constraints 102: Local Optima," \textit{Praxis Blog}, 2020. [9]

\bibitem{lmpc_thesis}
U. Rosolia, "Learning Model Predictive Control for Iterative Tasks," Ph.D. dissertation, University of California, Berkeley, 2018. [10]

\bibitem{soloperto2023safe}
R. Soloperto, P. N. A. Sopasakis, A. B. H. Weiss, P. Patrinos, and S. Grammatico, "Safe Exploration and Escape Local Minima with Model Predictive Control under Partially Unknown Constraints," \textit{IEEE Transactions on Automatic Control}, 2023. [11, 12]

\bibitem{traj_opt_local_optima}
A. Dragan, "Trajectory Optimization," \textit{CS287 Lecture Notes, UC Berkeley}, 2019. [13]

\bibitem{lattimore2020bandit}
T. Lattimore and C. Szepesvári, \textit{Bandit Algorithms}. Cambridge University Press, 2020. [14]

\bibitem{sutton2018reinforcement}
R. S. Sutton and A. G. Barto, \textit{Reinforcement Learning: An Introduction}, 2nd ed. MIT Press, 2018.

\bibitem{auer2002finite}
P. Auer, N. Cesa-Bianchi, and P. Fischer, "Finite-time analysis of the multiarmed bandit problem," \textit{Machine Learning}, vol. 47, no. 2-3, pp. 235-256, 2002.

\bibitem{koller2018learning}
T. Koller, F. Berkenkamp, M. Turchetta, and A. Krause, "Learning-based model predictive control for safe exploration," in \textit{2018 IEEE Conference on Decision and Control (CDC)}, 2018, pp. 6059-6066. [15, 16, 17]

\bibitem{alcan2022safe}
G. Alcan, "Safe Model Predictive Control," \textit{Aalto University Research}, 2022. [18]

\bibitem{berkenkamp2017safe}
F. Berkenkamp, M. Turchetta, A. P. Schoellig, and A. Krause, "Safe model-based reinforcement learning with stability guarantees," in \textit{Advances in Neural Information Processing Systems 30 (NIPS)}, 2017, pp. 908-918. [19]

\bibitem{hewing2020learning_rl}
L. Hewing, J. Kabzan, and M. N. Zeilinger, "Cautious model predictive control using Gaussian process regression," \textit{IEEE Transactions on Control Systems Technology}, vol. 28, no. 6, pp. 2736-2743, 2020. [20]

\bibitem{soloperto2023learning}
R. Soloperto, \textit{Learning-based Model Predictive Control with closed-loop guarantees}. Logos Verlag Berlin, 2023. [21]

\bibitem{bhattacharya2010search}
P. Bhattacharya, M. Likhachev, and V. Kumar, "Search-based path planning with homotopy class constraints," in \textit{Proceedings of the AAAI Conference on Artificial Intelligence}, vol. 24, no. 1, 2010. [22, 23]

\bibitem{bhattacharya2010search_aaai}
P. Bhattacharya, M. Likhachev, and V. Kumar, "Search-based Path Planning with Homotopy Class Constraints," \textit{AAAI}, 2010. [22]

\bibitem{hauser2008multi}
K. Hauser, T. Bretl, J. C. Latombe, and S. Rock, "Multi-modal motion planning in non-expansive spaces," \textit{The International Journal of Robotics Research}, vol. 27, no. 10, pp. 1143-1161, 2008. [24]

\bibitem{hauser2007multimodal}
K. Hauser and V. Ng-Thow-Hing, "Multi-modal motion planning for a humanoid robot," in \textit{2007 IEEE/RSJ International Conference on Intelligent Robots and Systems}, 2007, pp. 364-370. [25]

\bibitem{suh2020optimal}
J. Suh, A. D. Ames, and Y. Yue, "Optimal multi-modal locomotion via a shortest-path graph search on a learned value function," in \textit{2020 IEEE/RSJ International Conference on Intelligent Robots and Systems (IROS)}, 2020, pp. 8171-8178. [26]

\bibitem{ge2025learning}
Z. Ge, C. Chen, A. Sinha, and P. Varakantham, "On Learning Informative Trajectory Embeddings for Imitation, Classification and Regression," in \textit{AAMAS}, 2025. [27]

\bibitem{sung2012trajectory}
C. Sung, D. Feldman, and D. Rus, "Trajectory clustering for motion prediction," in \textit{2012 IEEE/RSJ International Conference on Intelligent Robots and Systems}, 2012, pp. 4813-4820. [28]

\bibitem{agrawal2012analysis}
S. Agrawal and N. Goyal, "Analysis of Thompson sampling for the multi-armed bandit problem," in \textit{Conference on Learning Theory}, 2012, pp. 39.1-39.26.

\bibitem{lai1985asymptotically}
T. L. Lai and H. Robbins, "Asymptotically efficient adaptive allocation rules," \textit{Advances in Applied Mathematics}, vol. 6, no. 1, pp. 4-22, 1985.

\bibitem{auer2010ucb}
P. Auer, "UCB revisited: Improved regret bounds for the stochastic multi-armed bandit problem," \textit{Bernoulli}, vol. 16, no. 1, pp. 1-28, 2010.

\bibitem{ortner2012refutable}
R. Ortner and D. Ryabko, "Optimism in the face of uncertainty should be refutable," \textit{arXiv preprint arXiv:1205.4820}, 2012. [29]

\bibitem{foster2020model}
D. J. Foster, A. Krishnamurthy, and H. Luo, "Open Problem: Model Selection for Contextual Bandits," in \textit{Conference on Learning Theory}, 2020, pp. 3842-3846. [30]

\bibitem{pacchiano2022model}
A. Pacchiano, "Model Selection for Contextual Bandits and Reinforcement Learning," \textit{Personal Blog}, 2022. [31]

\bibitem{agarwal2017corralling}
A. Agarwal, H. Luo, B. Neyshabur, and R. E. Schapire, "Corralling a band of bandit algorithms," in \textit{Conference on Learning Theory}, 2017, pp. 12-46.

\bibitem{chatterjee2020bandit}
S. Chatterjee, V. S. Borkar, and A. Sinha, "Bandit model selection in stochastic environments," in \textit{Advances in Neural Information Processing Systems 33 (NeurIPS)}, 2020. [32]

\bibitem{rosolia2017adaptive}
U. Rosolia and F. Borrelli, "An adaptive model predictive control framework for uncertain nonlinear systems," in \textit{2017 American Control Conference (ACC)}, 2017, pp. 451-456. [1]

\bibitem{maddalena2021kernel}
E. Maddalena, P. Scharnofske, and C. N. Jones, "Kernel Predictive Control: A Learning-Based MPC with Deterministic Guarantees," in \textit{Conference on Learning for Dynamics and Control}, 2021, pp. 104-115. [33]

\bibitem{wu2023learning}
S. P. Wu, "Learning Efficiently with Trajectory Data for Real World Robotics," Ph.D. dissertation, University of California, Berkeley, 2023. [34]

\bibitem{zheng2025novel}
J. Yan, Y. Wu, K. Ji, C. Cheng, and Y. Zheng, "A novel trajectory learning method for robotic arms based on Gaussian Mixture Model and k-value selection algorithm," \textit{PLoS ONE}, vol. 20, no. 2, p. e0318403, 2025. [35]

\bibitem{hasan2024trajectory}
M. A. Hasan, M. A. R. Ahad, S. Antol, and R. K. R. Ramadoss, "Trajectory Learning Using HMM: Towards Surgical Robotics Implementation," \textit{Sensors}, vol. 24, no. 10, p. 3179, 2024. [36, 37]

\bibitem{lee2007trajectory}
J.-G. Lee, J. Han, and K.-Y. Whang, "Trajectory clustering: a partition-and-group framework," in \textit{Proceedings of the 2007 ACM SIGMOD international conference on Management of data}, 2007, pp. 593-604.

\bibitem{yao2019survey}
D. Yao, C. Zhang, Z. Huang, J. Bi, and S. Li, "A survey on trajectory clustering analysis," \textit{Journal of Big Data}, vol. 6, no. 1, p. 104, 2019. [38]

\bibitem{calinon2007learning}
S. Calinon, F. Guenter, and A. Billard, "On learning, representing, and generalizing a task in a humanoid robot," \textit{IEEE Transactions on Systems, Man, and Cybernetics, Part B (Cybernetics)}, vol. 37, no. 2, pp. 286-298, 2007.

\end{thebibliography}

\vfill

\end{document}